\documentclass[12pt]{amsart}
\usepackage{a4wide}
\usepackage[latin9]{inputenc}
\usepackage{amsmath}
\usepackage{amsfonts}
\usepackage{amssymb}
\usepackage{amsthm}
\usepackage{subfigure}
\usepackage[usenames,dvipsnames]{color}
\usepackage{pstricks}
\usepackage{graphicx}
\usepackage[all]{xy}
\newtheorem{theo}{Theorem}[section]
\newtheorem{prop}[theo]{Proposition}
\newtheorem{coro}[theo]{Corollary}
\newtheorem{lemm}[theo]{Lemma}

\theoremstyle{definition}

\theoremstyle{remark}
\newtheorem{rema}[theo]{Remark}

\newcommand{\nwc}{\newcommand}
\nwc{\eps}{\epsilon}
\nwc{\ep}{\epsilon}
\nwc{\vareps}{\varepsilon}
\nwc{\Oph}{\operatorname{Op}_\hbar}
\nwc{\la}{\langle}
\nwc{\ra}{\rangle}

\nwc{\mf}{\mathbf} 
\nwc{\blds}{\boldsymbol} 
\nwc{\ml}{\mathcal} 

\nwc{\defeq}{\stackrel{\rm{def}}{=}}

\nwc{\cE}{\ml{E}}
\nwc{\cN}{\ml{N}}
\nwc{\cO}{\ml{O}}
\nwc{\cP}{\ml{P}}
\nwc{\cU}{\ml{U}}
\nwc{\cV}{\ml{V}}
\nwc{\cW}{\ml{W}}
\nwc{\tU}{\widetilde{U}}
\nwc{\IN}{\mathbb{N}}
\nwc{\IR}{\mathbb{R}}
\nwc{\IZ}{\mathbb{Z}}
\nwc{\IC}{\mathbb{C}}
\nwc{\IT}{\mathbb{T}}
\nwc{\IS}{\mathbb{S}}
\nwc{\tP}{\widetilde{P}}
\nwc{\tPi}{\widetilde{\Pi}}
\nwc{\tV}{\widetilde{V}}
\nwc{\supp}{\operatorname{supp}}
\nwc{\rest}{\restriction}
\nwc{\x}{\mathbf{x}}
\nwc{\y}{\mathbf{y}}
\nwc{\z}{\mathbf{z}}
\nwc{\w}{\mathbf{w}}

\begin{document}

\title[Spectral analysis of Morse-Smale flows II]{Spectral analysis of Morse-Smale flows II:\\
Resonances and resonant states}

\author[Nguyen Viet Dang]{Nguyen Viet Dang}

\address{Institut Camille Jordan (U.M.R. CNRS 5208), Universit\'e Claude Bernard Lyon 1, B\^atiment Braconnier, 43, boulevard du 11 novembre 1918, 
69622 Villeurbanne Cedex }

\email{dang@math.univ-lyon1.fr}

\author[Gabriel Rivi\`ere]{Gabriel Rivi\`ere}

\address{Laboratoire Paul Painlev\'e (U.M.R. CNRS 8524), U.F.R. de Math\'ematiques, Universit\'e Lille 1, 59655 Villeneuve d'Ascq Cedex, France}

\email{gabriel.riviere@math.univ-lille1.fr}

\begin{abstract} 
The goal of the present work is to compute
explicitely the correlation spectrum of a Morse-Smale flow in terms of the Lyapunov exponents of the Morse--Smale flow, 
the topology of the flow around periodic orbits and the monodromy of some given flat connection. The corresponding eigenvalues 
exhibit vertical bands when the flow has periodic orbits. As a corollary, we obtain sharp Weyl asymptotics for the dynamical resonances.
\end{abstract}

\maketitle

\section{Introduction}

Consider $p:\ml{E}\rightarrow M$ a smooth ($\ml{C}^{\infty}$) complex vector bundle of rank $N$ over a smooth, compact, oriented manifold $M$ without 
boundary of dimension $n$. 
Suppose now that $\ml{E}$ is endowed 
with a \emph{flat connection} $\nabla$ and that $V$ is a smooth vector field on $M$ which generates a flow $\varphi^t:M\rightarrow M$. For every 
$0\leq k\leq n$, $\varphi^t$ induces a flow $\Phi_k^t$ on the complex vector bundle $\mathbf{p}_k:\Lambda^k(T^*M)\otimes \ml{E}\rightarrow M$ whose sections
are differential $k$-forms on $M$ valued in sections of $\mathcal{E}$. 
The flow $\Phi^t_k$ satisfies the
equation 
$\varphi^t\circ\mathbf{p}_k=\mathbf{p}_k\circ \Phi_k^t$. In this article, we aim 
at describing the long time behaviour of this family of induced flows for vector fields enjoying some 
hyperbolic features. For this purpose, it is natural to look at the action of these flows on smooth sections 
of $\Lambda^k(T^*M)\otimes \ml{E}$, i.e. given $\psi_0$ in $\Omega^k(M,\ml{E})$, one would like to understand the asymptotic behaviour of
\begin{equation}\label{e:pullback}\Phi_k^{-t*}\left(\psi_0\right)\end{equation}
as $t\rightarrow+\infty$. Such quantities are sometimes referred as linear cocycles in the literature from dynamical systems. Actually, it is convenient to observe that $\Phi_k^{-t*}\left(\psi_0\right)$ solves the partial differential equation
$$\boxed{\partial_t\psi=-\ml{L}_{V,\nabla}\psi,\quad\psi(t=0)=\psi_0,}$$
where $$\ml{L}_{V,\nabla}:=(d^{\nabla}+\iota_V)^2,$$
with $d^{\nabla}$ the coboundary operator\footnote{Recall that $d^{\nabla}\circ d^{\nabla}=0$ as the connection is flat.} 
induced by $\nabla$ and $\iota_V$ the contraction by the vector field $V$. Hence, one would be able to 
determine the limit of~\eqref{e:pullback} as soon as one has built a good spectral theory for the Lie derivative $\ml{L}_{V,\nabla}$. 

In recent 
years, many progresses have been made towards the construction of functional frameworks adapted to smooth vector fields satisfying 
certain hyperbolicity assumptions. For instance, in the case of Anosov vector fields and in the case of the trivial vector bundle $M\times\IC$, 
Butterley and Liverani constructed Banach spaces for which the Lie derivative has good spectral properties such as discrete 
spectrum~\cite{BuLi07}. This was extended to more general vector bundles by Giulietti, Liverani and Pollicott~\cite{GiLiPo13} and 
applied to prove the meromorphic continuation of the Ruelle zeta function. The result of Butterley and 
Liverani was recovered by Faure and Sj\"ostrand~\cite{faure2011upper} via microlocal 
techniques inspired by the study of resonances of semiclassical Schr\"odinger 
operators~\cite{HeSj86, dyatlov2016ruelle}. These microlocal methods were
then extended to the situation of vector bundles by Dyatlov and 
Zworski~\cite{dyatlov2016dynamical} 
who gave a microlocal proof of the meromorphic continuation of the Ruelle zeta function. 
Still from a 
microlocal perspective, we can also mention the works of Tsujii~\cite{Ts12} and Faure--Tsujii~\cite{faure2013band, FaTs13} 
which are based on the use of the FBI 
transform. Beyond the Anosov case, Dyatlov and Guillarmou proved that a similar 
microlocal approach can be performed for Axiom A flows~\cite{smale1967differentiable} 
at the expense of making some restriction on the 
supports of the sections $\psi_0$ in $\Omega^k(M,\ml{E})$~\cite{dyatlov2016pollicott}. More specifically, for 
Axiom A flows, 
Smale proved that there exists  
a decomposition of the nonwandering set of the flow into finitely many basic sets $(\Lambda_i)_{i=1,\ldots, K}$ and the results of 
Dyatlov and Guillarmou hold \emph{locally} in a neighborhood of some fixed
given basic set. Recently, we showed how to construct a proper \emph{global} 
spectral theory for certain families of Axiom $A$ flows, namely Morse-Smale flows~\cite{dang2016spectral, dangrivieremorsesmale1}. This was 
achieved by combining some ideas from dynamical systems going back to the original works of Smale~\cite{smale2000morse} with the microlocal 
approach of Faure and Sj\"ostrand. Even if we focus on the case of flows, as in the present article, we emphasize that results 
on flows follow 
from many progresses that have been made towards the understanding of hyperbolic diffeomorphisms and we refer to the recent book of Baladi for 
a detailed account of this case~\cite{Ba16}.

\subsection{Dynamical correlations and their Laplace transform} As expected, these spectral results have nice dynamical consequences which can 
be formulated in terms of the correlation function:
\begin{eqnarray*}
\boxed{C_{\psi_1,\psi_2}(t)=\int_M \psi_1\wedge \Phi_k^{-t*}(\psi_2),}
\end{eqnarray*}
where $\psi_1\in \Omega^{n-k}(M,\cE')$, $\psi_2\in \Omega^{k}(M,\cE)$ and $\cE^\prime$ is the dual
bundle of $\cE$. Following the works of Pollicott~\cite{Po85} and 
Ruelle~\cite{Ru87a}, we can introduce its Laplace transform~:
\begin{eqnarray*}
\boxed{\widehat{C}_{\psi_1,\psi_2}(z)=\int_0^\infty C_{\psi_1,\psi_2}(t)e^{-tz}dt.}
\end{eqnarray*}
Note that this function is well defined for $\text{Re}(z)>0$ large enough and most of the above mentionned results show that it 
has a meromorphic extension to the entire complex plane\footnote{In the Axiom A case, one has to consider $\psi_1$ and $\psi_2$ 
compactly supported near a fixed basic set.}. The poles and residues of this function describe in some sense the fine structure 
of the long time dynamics of the flow $\Phi_k^{-t*}$. For instance, let us state the precise result 
when $V$ is a Morse-Smale flow which is~$\ml{C}^1$-linearizable~\cite{dangrivieremorsesmale1} -- see section~\ref{s:morse-smale}. In that 
framework, there exists a minimal discrete\footnote{We mean that it has no accumulation point. In particular, it is at most countable.} set 
$\ml{R}_k(V,\nabla)\subset\IC$ such that, given any $(\psi_1,\psi_2)\in\Omega^{n-k}(M,\ml{E}')\times\Omega^{k}(M,\ml{E})$, the map 
$z\mapsto \hat{C}_{\psi_1,\psi_2}(z)$ has a meromorphic extension whose poles are contained inside $\ml{R}_k(V,\nabla)$. 
These poles are called the \emph{Pollicott-Ruelle resonances}. Moreover, given any such 
 $z_0\in\ml{R}_k(V,\nabla)$, there exists an integer $m_k(z_0)$ and a linear map of \emph{finite rank}
\begin{equation}\label{e:spectral-projector-intro}\pi_{z_0}^{(k)}:\Omega^k(M,\ml{E})\rightarrow\ml{D}^{\prime k}(M,\ml{E})\end{equation}
 such that, given any $(\psi_1,\psi_2)\in\Omega^{n-k}(M,\ml{E}')\times\Omega^{k}(M,\ml{E})$, one has, in a small neighborhood of $z_0$,
 $$\hat{C}_{\psi_1,\psi_2}(z)=\sum_{l=1}^{m_k(z_0)}(-1)^{l-1}\frac{\left\la(\ml{L}_{V,\nabla}^{(k)}+z_0)^{l-1}\pi_{z_0}^{(k)}(\psi_2),\psi_1\right\ra}{(z-z_0)^l}+R_{\psi_1,\psi_2}(z),$$
 where $R_{\psi_1,\psi_2}(z)$ is a holomorphic function. Here, we use the convention that $\ml{D}^{\prime k}(M,\ml{E})$ represents 
 the currents of degree $k$ with values in $\ml{E}$. Elements in the range of $\pi_{z_0}^{(k)}$ are called 
 \emph{Pollicott-Ruelle resonant states} and they can be interpreted as the generalized eigenvectors 
 of the operator $-\ml{L}_{V,\nabla}^{(k)}$ acting on an appropriate Sobolev space~\cite{dangrivieremorsesmale1}. In particular, the dimension of the 
 range of $\pi_{z_0}^{(k)}$ is the algebraic multiplicity of the resonance $z_0$ 
 viewed as an eigenvalue of this operator.

\subsection{Quantum chaos, instantonic theories and Epstein--Glaser renormalization}
Recall that, beyond the dynamical aspects, one of the motivation of Faure and collaborators, as we understand it, comes from 
quantum chaos. For geodesic flows in negative curvature\footnote{These are the simplest examples of contact Anosov flows.}, 
$\Phi^{-t*}(\psi)$ converges to some equilibrium state $\psi_\infty$ at the limit when $t\rightarrow +\infty$. 
The Pollicott-Ruelle spectrum describes the fluctuations around the equilibrium $\psi_\infty$, and 
one hopes that, for general geodesic flows in nonconstant negative curvature, the emergent dynamics induced by the vector field 
is a model of \emph{quantum chaos}. Before stating our main results on this correlation spectrum for Morse-Smale flow, we would like 
to present another motivation coming from mathematical physics and emphasize possible links of these dynamical problems with the work of 
Frenkel--Losev--Nekrasov on instantonic quantum field theories.

Instantons arise in mathematical physics as critical points of some natural
variational problems. In the context of Morse theory,
the term \emph{instanton} denotes 
gradient lines connecting two critical points of the vector field $V$.  
In several geometrical problems, it is often proved that the moduli space of instantons
is finite dimensional in some sense.
In~\cite{frenkel2011instantons,frenkel2008instantons,frenkel2007notes,losev2011new}, 
Frenkel--Losev--Nekrasov started an ambitious program of constructing quantum field theories
by integrating on finite dimensional moduli space of instantons. 
Since this is a very hard problem in general, they propose to start by investigating 
$1$-dimensional QFT which is just instantonic quantum mechanics.
This is a version of supersymmetric quantum mechanics
where the vector field $V$ plays the role
of the generalized Laplacian from Hodge theory.
The Hamiltonian of the theory is a multiple of the Lie derivative $\mathbb{H}=i\mathcal{L}_V$.
A natural problem in quantum theory is to specify all eigenvalues and eigenstates of the Hamiltonian $\mathbb{H}$.
For general Hamiltonians this is usually very hard.
However, for $\mathbb{H}=i\mathcal{L}_{V}$, it is hoped that the
system is integrable and actually, it is showed in~\cite[paragraph 3.7 p.~508]{frenkel2011instantons} that
for the height function on the Riemann sphere $\mathbb{CP}^1$, the spectrum of $\mathbb{H}$ coincides with
$i\mathbb{N}$ and the structure of the eigenstates is completely understood by a local construction near the south pole
and by a procedure of extension of distributions near the north pole~\cite[p.~516-517]{frenkel2011instantons} 
in the style of Epstein--Glaser renormalization in quantum field theory~\cite{Epstein}. We extended this explicit description
to more general gradient flows in~\cite{dang2016spectral}. One can now briefly illustrate the process describing the passage from 
quantum mechanics to instantonic theories in the \textbf{specific case of Morse-Smale gradient flows}~\cite[section 5]{enciso2008morse},  
which can be thought of as the inverse of the path followed by Faure et al. for Anosov flows.
One starts from a rescaled Witten Laplacian which
is a deformation of the Hodge Laplacian by a Morse potential $f$~:
$\mathbb{H}_{\hbar}=\frac{1}{2}\left(-\hbar\Delta+\hbar^{-1}\vert df\vert^2+\mathcal{L}_{V}+\mathcal{L}^*_{V}\right)$.
Then one can wonder what the eigenvalues of this new quantum Hamiltonian are. For that purpose, we conjugate $\mathbb{H}$ with $e^{\frac{f}{\hbar}}$ which yields 
the new Hamiltonian~:
$
\tilde{\mathbb{H}}_{\hbar}=e^{\frac{f}{\hbar}}\mathbb{H}_{\hbar}e^{-\frac{f}{\hbar}} = \mathcal{L}_{V}-\frac{\hbar\Delta}{2}.
$
At the \emph{instantonic limit} $\hbar\rightarrow 0$, $\tilde{\mathbb{H}}_{\hbar}$ becomes formally
$\tilde{\mathbb{H}}_0=\mathcal{L}_{V}$ which is nothing else but the Lie derivative along the gradient vector field whose spectrum 
was computed explicitely in~\cite{dang2016spectral} or can be deduced from the upcoming results applied in the particular 
case of Morse-Smale gradient flows. Therefore, the spectrum of the Witten Laplacian converges to the spectrum 
of the Lie derivative.

\section{Statement of the main results}

The main objective of this article is to give a complete description of the Pollicott-Ruelle resonances \emph{and} resonant states in the case of Morse-Smale vector fields satisfying 
certain generic linearizing assumptions. Let us start with the resonances which are simpler to describe. Recall that a Morse-Smale flow $\varphi^t$ is a flow whose nonwandering set is the union 
of finitely many hyperbolic closed orbits and hyperbolic fixed points that we denote by $\Lambda_1,\ldots,\Lambda_K$. These are called the critical elements (or the basic sets) of the flow. 
To each $\Lambda_j$ is associated an unstable (resp. stable) manifold 
$W^u(\Lambda_j)$ (resp. $W^s(\Lambda_j)$). These unstable manifolds form a partition of $M$ and, by definition of a Morse-Smale flow, they enjoy some transversality properties -- see 
paragraph~\ref{s:morse-smale} for a brief reminder. Such unstable manifolds can be either orientable or 
not\footnote{Non orientability can only occurs for closed orbits.}. We define the twisting index of $\Lambda$ as 
$$\varepsilon_{\Lambda}=0\ \text{if $W^u(\Lambda)$ is orientable, and}\ \varepsilon_{\Lambda}=\frac{1}{2}\ \text{otherwise}.$$
To every closed orbit $\Lambda$, we also associate a positive number $\ml{P}_{\Lambda}$ which is the 
minimal period of the closed orbit and an element $M_{\ml{E}}(\Lambda)$ which is a monodromy matrix for the parallel transport around 
$\Lambda$. We denote the eigenvalues of this matrix by $(e^{2i\pi\gamma_j^{\Lambda}})_{j=1,\ldots, N}$ where $\gamma_j^{\Lambda}$ are complex numbers. For every fixed point $\Lambda$, we define
$$\sigma_{\Lambda}:=\{0\},$$
and the multiplicity of $0$ is defined as $\mu_{\Lambda}(0)=N$. For every closed orbit, we set
$$\sigma_{\Lambda}:=\left\{-\frac{2i\pi(m+\varepsilon_\Lambda+\gamma_j^{\Lambda})}{\mathcal{P}_\Lambda}:1\leq j\leq N,\ m\in\IZ \right\},$$
and the multiplicity of $z_0$ in $\sigma_{\Lambda}$ is given by
$$\mu_{\Lambda}(z_0):=\left|\left\{(j,m):z_0=-\frac{2i\pi(m+\varepsilon_\Lambda+\gamma_j^{\Lambda})}{\mathcal{P}_\Lambda}\right\}\right|.$$

\subsection{Resonances on the imaginary axis} Our first main result describes the resonances lying on the imaginary axis in the case 
where the flat connection $\nabla$ preserves some Hermitian structure on the fibers of $\ml{E}$. 
Recall that this is equivalent to fixing some unitary representation of the fundamental group $\pi_1(M)$~\cite[Th.~13.2]{taubes2011differential}. 
In that case, all the $\gamma_j^{\Lambda}$ are real and we show:
\\
\\
\fbox{
\begin{minipage}{0.94\textwidth}  
\begin{theo}\label{mainthmintro}
 Let $\ml{E}\rightarrow M$ be a smooth, complex, hermitian 
vector bundle of dimension $N$ endowed with a flat unitary connection $\nabla$. 
Suppose that $V$ is a Morse-Smale vector field which is $\ml{C}^{\infty}$-diagonalizable. 

Then, for every $0\leqslant k\leqslant n$, 
$$\ml{R}_k(V,\nabla)\subset\{z:\operatorname{Re}(z)\leq 0\},$$ 
and one has
\begin{eqnarray*}
\ml{R}_k(V,\nabla)\cap i\IR =\bigcup_{\Lambda\ \text{fixed point}:\ \operatorname{dim}\ W^s(\Lambda)=k}\sigma_{\Lambda}\cup\bigcup_{\Lambda\ \text{closed orbit}:\ \operatorname{dim}\ W^s(\Lambda)\in\{k,k+1\}}\sigma_{\Lambda},
\end{eqnarray*}
where a resonance $z_0\in\ml{R}_k(V,\nabla)\cap i\IR$ appears with the algebraic multiplicity
$$\sum_{\Lambda\ \text{fixed point}:\ \operatorname{dim}\ W^s(\Lambda)=k}\mu_{\Lambda}(z_0)
+\sum_{\Lambda\ \text{closed orbit}:\ \operatorname{dim}\ W^s(\Lambda)\in\{k,k+1\}}\mu_{\Lambda}(z_0).$$
\end{theo}
\end{minipage}
}
\\
\\
Before making some comments on this first result, let us observe that we made an extra assumption on the vector field saying that it is $\ml{C}^{\infty}$-diagonalizable -- 
see section~\ref{s:morse-smale} for the precise definition. It roughly means that we can linearize the vector field in a smooth chart near any $\Lambda_j$. Thanks to the 
Sternberg-Chen Theorem, this is satisfied as soon as certain nonresonance assumptions are satisfied~\cite{nelson2015topics, WWL08}. 
Combined with the classical results of Peixoto~\cite{Pe62} and Palis~\cite{Pa68}, we can then verify that this assumption is in some sense generic among Morse-Smale vector 
fields: we refer to paragraph~\ref{s:morse-smale} and to the discussion following Theorem~\ref{t:maintheo-full} for more precision on this hypothesis. 
One of the reason for this assumption is that we do not only aim at describing these dynamical eigenvalues but also their 
corresponding generalized eigenmodes -- see Theorem~\ref{t:maintheo-full} for a complete statement. Besides the fact that the question of describing 
the eigenstates is interesting on its own, we shall also see in~\cite{dangrivieremorsesmale3} that these eigenmodes have 
interesting topological properties related to Morse inequalities and Reidemeister torsion. If we were only interested in the eigenvalues, 
it would probably be sufficient to relate these eigenvalues to the zeros of some dynamical zeta function via a trace formula as 
in~\cite{BaTs08, Ba16} and maybe avoid some of the linearization assumptions we have. Yet, it would not be necessarily much simpler from the technical 
point of view as this would require to justify such a formula in our context. This would mean to use the notion of 
distributional traces~\cite{GuSt90} of a flow with fixed points and closed orbits. Our somewhat direct 
approach avoid this difficulty and, along the way, it also gives in addition a complete description of the generalized eigenmodes. Coming back to the nonresonance assumption, one should probably be able to lower the regularity assumptions (at 
the expense of a slightly more technical work), and still describe the Pollicott-Ruelle resonant states explicitely in some half plane 
$\{z:\operatorname{Re}(z)\geq -T\},$ where $T$ would depend on the 
regularity assumptions. Yet, removing these assumptions is beyond the scope of the present article 
which essentially aims at giving examples where one can explicitely determine the correlation spectrum.

Let us now briefly comment our first Theorem. First of all, it completely determines the Pollicott-Ruelle resonances on the imaginary axis 
in terms of the periods of the flow, of the topology of the unstable manifolds and of the monodromy around every closed orbit. 
In particular, up to the periods, this part of the spectrum is completely determined by the ``topology'' of the flow. 
We also emphasize that eigenvalues in every degree are associated with Smale's partition of $M$ into stable 
manifolds~\cite{smale2000morse}. We remark that, even if many progresses were made towards understanding 
the Pollicott-Ruelle spectrum, there are not so many examples where one can compute the spectrum explicitely. In the case of maps, we can mention 
the case of hyperbolic linear automorphisms of the torus where there is in fact only one resonance~\cite{BlKeLi02}, the case of an hyperbolic fixed point 
which can be derived from~\cite{BaTs08} (see also~\cite{FaTs15}) or the one of analytic expanding circle maps arising from finite Blaschke 
products~\cite{BanSlJu15, bandtlow2016lower}. In the case of geodesic flows on hyperbolic manifolds, these resonances were shown to 
be in correspondance with the spectrum of the Laplace Beltrami operator~\cite{dyatlov2015power, guillarmou2016classical} -- see also~\cite{FlFo03} 
for earlier related results. In the case of Morse-Smale gradient flows, we gave a complete description of the correlation spectrum. 
The main differences with that last reference are the addition of the flat connection and the presence of closed orbits. These closed 
orbits are in fact responsible for the vertical lines we can observe inside $\ml{R}_k(V,\nabla)$. As a first corollary of our analysis, 
let us point the following Weyl formula:
\begin{coro}\label{c:fractalweylaw} Suppose the assumptions of Theorem~\ref{mainthmintro} are satisfied and let $0\leq k\leq n$. Then, as $T\rightarrow+\infty$, one has
 $$\left|\left\{z_0\in\ml{R}_k(V,\nabla)\cap i\IR: |\operatorname{Im}(z_0)|\leq T\right\}\right|=\frac{N T}{\pi}\left(\sum_{\Lambda\ \text{closed orbit}:\ \operatorname{dim}\ W^s(\Lambda)\in\{k,k+1\}}\ml{P}_{\Lambda}\right)+\ml{O}(1),$$
 where the resonances are counted with their algebraic multiplicity.
\end{coro}
Recall that, for Anosov flows, Faure and Sj\"ostrand proved that the resonances near the imaginary axis verifies Weyl's upper bound 
in the limit $\text{Im}(z)\rightarrow+\infty$~\cite{faure2011upper} -- see also~\cite{DatDyZw14, FaTs13} in the contact case.

\subsection{Pollicott-Ruelle resonances and Weyl's law}

Our analysis will in fact give an explicit description of the full correlation spectrum inside $\{z:\operatorname{Re}(z)\leq 0\}$ 
in terms of the Lyapunov exponents of the flow, of the periods of the flow, of the topology of the unstable manifolds and 
of the monodromy around every closed orbits. Since this description is a little bit combinatorial and, for the simplicity of exposition, 
we just mention the following consequence of Theorem~\ref{t:maintheo-full}:
\\
\\
\fbox{
\begin{minipage}{0.94\textwidth}  
\begin{theo}\label{mainthmintro2}
 Suppose the assumptions of Theorem~\ref{mainthmintro} are satisfied. Then, for every $0\leqslant k\leqslant n$ and for every critical
element $\Lambda$, there exists
a sequence of complex numbers $ \left(z_{\Lambda,k}(j)\right)_{j\geq 1} $
such that $$\forall j\geq 1,\ \operatorname{Re}\left(z_{\Lambda,k}(j)\right)\leqslant 0,\ \lim_{j\rightarrow +\infty}
z_{\Lambda,k}(j)=-\infty,$$ 
and
$$\ml{R}_k(V,\nabla)= \bigcup_{\Lambda}\bigcup_{j\geq 1} \left(z_{\Lambda,k}(j)+\sigma_\Lambda\right).$$ 
\end{theo}
\end{minipage}
}
\\
\\
For every closed orbit or critical point $\Lambda$, the $z_{\Lambda,k}(j)$ can be determined explicitely as linear (integer) combination 
of the eigenvalues of the linearized system near $\Lambda$. In some sense, they will only depend on the local properties of the flow near $\Lambda$.
For simplification, we also made the assumption that $\nabla$ preserves a smooth hermitian structure but our analysis remains true under the slightly 
more general hypothesis that $M_{\ml{E}}(\Lambda)$ is diagonalizable for every closed orbit. 
In that case, it may happen that there are finitely many bands of resonances on the half plane $\{z:\operatorname{Re}(z)> 0\}$. This Theorem should 
be compared with the results of Faure and Tsujii~\cite{faure2013band, FaTs13} on the Pollicott-Ruelle spectrum of contact 
Anosov flows. In that framework, they proved that the resonances exhibit in the limit $\text{Im}(z)\rightarrow+\infty$ 
a band structure which is completely determined by the unstable Jacobian of the flow. Here, our analysis show that this 
band structure remains true for Morse-Smale flows and it is in fact given by vertical lines of resonances which can be explicitely 
determined. Finally, to every critical element $\Lambda$ of the Morse--Smale flow, we associate
a \textbf{convex polytope} $\mathcal{Q}_\Lambda$ in $\mathbb{R}^n$ which depends only on the
eigenvalues of the linearization of the vector field $V$ near $\Lambda$. These convex polytopes are explicitely defined by~\eqref{e:polytope-fixed-point} 
and~\eqref{e:polytope-closed-orbit}. Our last Theorem gives Weyl's law satisfied by elements in $\ml{R}_k(V,\nabla)$:
\\
\\
\fbox{
\begin{minipage}{0.94\textwidth}  
\begin{theo}[Weyl's law] Suppose that the assumptions of Theorem~\ref{mainthmintro} are satisfied. 
For every $0\leq k\leq n$, define 
$$N_k(T):=\left|\left\{z_0\in\ml{R}_k(V,\nabla): \vert \operatorname{Im}(z_0)\vert\leqslant T\ \text{and}\ \operatorname{Re}(z_0)\geqslant -T\right\}\right|,$$
where the Pollicott-Ruelle resonances are counted with their algebraic multiplicity. Then, one has, as $T\rightarrow+\infty$,
\begin{eqnarray*} N_k(T) 
&=& \frac{N n!}{k! (n-k)!}\left(\sum_{j=1}^K\operatorname{Vol}_{\IR^n}(\ml{Q}_{\Lambda_j})\right)T^n+\ml{O}(T^{n-1}).
\end{eqnarray*}
\end{theo}
\end{minipage}
}
\\
\\
Except for the case of Morse-Smale gradient flows~\cite{dang2016spectral}, we are not aware of the existence of asymptotic formulas for the 
counting function of Pollicott-Ruelle resonances. In the case of the geodesic vector field on hyperbolic manifold, they can maybe be derived 
following~\cite{dyatlov2015power} but it is not completely obvious what the exponent would be for $T$. Again, this Theorem is just a Corollary 
of the analysis we will perform in this article and we will in fact be able to determine exactly which resonances are inside 
these large boxes.

\subsection{Constructing the resonant states}
As was already explained, not only we will describe the eigenvalues and their multiplicity but also their corresponding generalized eigenmodes. 
In fact, the way we prove the above results rely on our explicit construction of the generalized eigenmodes of the operator 
$-\mathcal{L}_{V,\nabla}$ acting 
on the anisotropic Sobolev spaces we have defined in~\cite{dangrivieremorsesmale1}. More precisely, given any critical element $\Lambda$ and any $z_0\in \IC$,
we will first consider \textbf{germs of solution $u\in \mathcal{D}^{\prime,\bullet}(\cE)$} of the equation
\begin{equation}\label{e:eigenequationintro}
\boxed{ \left(\mathcal{L}_{V,\nabla}+z_0\right)\psi_0=0}
\end{equation}
whose support near $\Lambda$ is contained in the unstable manifold $W^u(\Lambda)$ (see equation~\eqref{e:constraint-support}) and whose wavefront set 
lies in the union of conormals of strongly unstable manifolds of $\Lambda$ (see equation~\eqref{e:constraint-wavefront} for a precise statement). We will solve 
this eigenvalue problem explicitely near $\Lambda$ and show that it imposes restriction on the possible values of $z_0$. Then we will try to extend these 
local solutions into currents which are globally defined on $M$. What 
we would do intuitively is to use the fact that $\Phi^{-t*}(\psi_0)=e^{tz_0} \psi_0$ to propagate the
local solution to define some current supported on $\overline{W^u(\Lambda)}$. Yet, this would require to analyze carefully the closure 
of $W^u(\Lambda)$ and this may turn to be a delicate task -- see~\cite{Lau92} for related problems in the case 
of Morse-Smale gradient flows. Instead of that, we will make use of the spectral projectors $\pi_{z_0}^{(\bullet)}$ which are given to us by our spectral 
analysis in~\cite{dangrivieremorsesmale1}. In that manner, we hide the difficulty of understanding the dynamics of the flow near the boundary 
of $W^u(\Lambda)$ into the construction of our anisotropic Sobolev space -- see for instance~\cite[Sect.~4]{dangrivieremorsesmale1} for results related to 
this delicate dynamical issue. Applying the spectral projectors to the locally defined currents allow to extend them into globally defined currents. In some 
sense, the use of spectral theory elegantly replaces the method of Epstein--Glaser renormalization used in~\cite{frenkel2008instantons,frenkel2011instantons} 
to extend distributions -- see paragraph~\ref{sss:Epstein-Glaser} especially Theorem~\ref{t:epstein-glaser} 
for a more detailed discussion. We also note that the generalized eigenmodes we will construct 
are related to the so-called Ruelle-Sullivan currents~\cite{RuSu75} and we shall come back on that issue in paragraph~\ref{sss:Ruelle-Sullivan}. 

If we come back to our problem, the price we 
pay when we extend these currents is that they are 
not a priori true eigenmode and they only satisfy the generalized eigenvalue equation:
$$\boxed{\left(\mathcal{L}_{V,\nabla}+z\right)^{m_{\bullet}(z_0)}\overline{\psi}_0=0.}$$
Once these generalized eigenmodes are constructed, we will make use of the Morse-Smale dynamics and of the spectral analysis from~\cite{dangrivieremorsesmale1} 
to conclude that they indeed generate all the Pollicott-Ruelle resonant states. The main result on that question is Theorem~\ref{t:maintheo-full} which is in 
fact the main result of the present article.

Finally, we briefly mention that the resonant states associated to resonances lying on the imaginary axis have a strong topological meaning. We shall come back 
to this issue in the companion article~\cite{dangrivieremorsesmale3}.




\subsection{Organization of the article} In section~\ref{s:morse-smale}, we start with a brief reminder on Morse-Smale flows and on the spectral results 
from~\cite{dangrivieremorsesmale1}. For more details on both issues, we refer the reader to this reference. In section~\ref{s:connection}, we show 
how we can remove the contribution coming from the flat connection $\nabla$ by shifting the spectrum. After that, we show in 
section~\ref{s:scalar} how to reduce the question of solving the eigenvalue equation to a scalar problem. In this section, we also introduce 
several conventions related to the Pollicott-Ruelle spectrum. In section~\ref{s:local-form-eigenmode}, we solve the eigenvalue problem locally near every 
critical element of the vector field. Finally, in section~\ref{s:proof}, we construct the resonant states and gives our main Theorem from which all 
the results of this section follows. In the appendix, we briefly review some facts from Floquet theory that we extensively use all along the article.

\subsection*{Acknowledgements} We warmly thank Fr\'ed\'eric Faure for many explanations on his works with Johannes Sj\"ostrand and Masato Tsujii. We also acknowledge 
useful discussions related to this article and its companion articles~\cite{dangrivieremorsesmale1, dangrivieremorsesmale3} 
with Livio Flaminio, Colin Guillarmou, Benoit Merlet, Fr\'ed\'eric Naud and Patrick Popescu Pampu. The second author 
is partially supported by the Agence Nationale de la Recherche through the Labex CEMPI (ANR-11-LABX-0007-01) and the 
ANR project GERASIC (ANR-13-BS01-0007-01).

\section{A brief reminder on Morse-Smale flows and on~\cite{dangrivieremorsesmale1}}\label{s:morse-smale}

\subsection{Morse-Smale flows}

We say that $\Lambda\subset M$ is an elementary critical element if $\Lambda$ is either a fixed point or a closed orbit of $\varphi^t$. 
Such an element is said to be hyperbolic if the fixed point or the closed orbit is hyperbolic -- see appendix A of~\cite{dangrivieremorsesmale1} for a brief reminder. 
Following~\cite[p.~798]{smale1967differentiable} , $\varphi^t$ is \textbf{a Morse-Smale flow} if the following properties hold:
\begin{enumerate}
 \item the non-wandering set $\operatorname{NW}(\varphi^t)$ is the union of finitely many elementary critical elements 
$\Lambda_1,\ldots,\Lambda_K$ 
which are hyperbolic,
 \item for every $i,j$ and for every $x$ in $W^u(\Lambda_j)\cap W^s(\Lambda_i)$, one has
 \footnote{See appendix of~\cite{dangrivieremorsesmale1} for the precise definition of the stable/unstable manifolds $W^{s/u}(\Lambda)$.} 
 $T_xM=T_xW^u(\Lambda_j)+T_xW^s(\Lambda_i)$.
\end{enumerate}
 We now briefly expose some important properties of Morse-Smale flows and we refer to~\cite{dangrivieremorsesmale1} for a more 
detailed exposition on the dynamical properties of these flows.
Under such assumptions, one can show that, for every $x$ in $M$, there exists an unique couple $(i,j)$ such 
that $x\in W^u(\Lambda_j)\cap W^s(\Lambda_i)$ (see e.g.~\cite[Lemma 3.1]{dangrivieremorsesmale1}). 
In particular, \emph{the unstable manifolds $(W^u(\Lambda_j))_{j=1,\ldots, K}$ form a partition of $M$,} i.e.
$$M=\bigcup_{j=1}^KW^u(\Lambda_j),\ \ \text{and}\ \ \forall i\neq j,\ W^u(\Lambda_i)\cap W^u(\Lambda_j)=\emptyset.$$
The same of course holds for stable manifolds. One of the main feature of such flows is the following result which is due to Smale~\cite{dangrivieremorsesmale1,smale2000morse}:
\begin{theo}[Smale]\label{t:smale} Suppose that $\varphi^t$ is a Morse-Smale flow. Then, for every $1\leq j\leq K$, the closure of $W^u(\Lambda_j)$ is the union of certain $W^u(\Lambda_{j'})$. 
Moreover, we say that $W^u(\Lambda_{j'})\leqq 
 W^u(\Lambda_j)$ if $W^u(\Lambda_{j'})$ is contained in the closure of $W^u(\Lambda_{j})$, then, $\leqq$ is a partial ordering.
Finally if $W^u(\Lambda_{j'})\leqq 
W^u(\Lambda_j)$, then $\operatorname{dim}W^u(\Lambda_{j'})\leq \operatorname{dim}W^u(\Lambda_{j}).$
\end{theo}
The partial order relation on the collection of subsets $W^u(\Lambda_j)_{j=1}^K$ defined above is called \textbf{Smale causality relation}. Following Smale, 
we define an oriented graph\footnote{This diagram is the Hasse diagram associated to the poset $\left(W^u(\Lambda_j)_{j=1}^K,\leqq\right)$.} 
$D$ called \emph{Smale quiver} whose $K$ vertices are given by $W^u(\Lambda_j)_{j=1}^K$. Two vertices $W^u(\Lambda_j), W^u(\Lambda_i) $ are connected
by an oriented path starting at $W^u(\Lambda_j)$ and ending at $W^u(\Lambda_i)$ iff $W^u(\Lambda_j)\geqq W^u(\Lambda_i)$. From the works of 
Peixoto~\cite{Pe62}, it is known that Morse-Smale flows form an open and dense subset of all smooth vector fields in dimension $2$ while Palis 
showed that in higher dimension they form an open subset~\cite{Pa68}. In particular, if we perturb a little bit a Morse-Smale flow, it 
remains Morse-Smale. For our analysis, we will need to choose Morse-Smale flows satisfying certain \emph{generic} assumptions on their linearization 
near critical elements as we shall now explain. It would be interesting to remove these assumptions and we refer to the discussion following our 
main Theorem~\ref{t:maintheo-full} for more details on that question.

More precisely, our linearizing assumptions are as follows. We fix $1\leq k\leq \infty$ and we say that the Morse-Smale flow is $\ml{C}^k$-linearizable if for every $1\leq i\leq K$, the following hold:
\begin{itemize}
 \item If $\Lambda_i$ is a fixed point, there exists a $\ml{C}^k$ diffeomorphism 
$h: B_n(0,r)\rightarrow W$ (where $W$ is a small open neighborhood of $\Lambda_i$ and $B_n(0,r)$ is a small ball of radius $r$ 
centered at $0$ in $\IR^n$) and a linear map $A_i$ on $\IR^n$ such that $V\circ h=dh\circ L$ where $V$ is the vector field generating $\varphi^t$ and where
$$L(x)=A_ix.\partial_x.$$
 \item If $\Lambda_i$ is a closed orbit of period $\ml{P}_{\Lambda_i}$, there exists a $\ml{C}^k$ diffeomorphism 
$h: B_{n-1}(0,r)\times\IR/(\ml{P}_{\Lambda_i}\IZ)\rightarrow W$ (where $W$ is a small open neighborhood of $\Lambda_i$ and $r>0$ is small) and a 
smooth map $A:\IR/(\ml{P}_{\Lambda_i}\IZ)\rightarrow M_{n-1}(\IR)$ such that $V\circ h=dh\circ L$ with
$$L_i(x,\theta)=A_i(\theta)x.\partial_x+\partial_{\theta}.$$
\end{itemize}
In other words, the flow can be put into a normal form in a certain chart of class $\ml{C}^k$. We shall say that \textbf{a Morse-Smale flow is $\ml{C}^k$-diagonalizable} 
if it is $\ml{C}^k$-linearizable and if, for every critical element $\Lambda$, either the linearized matrix $A\in GL_n(\IR)$ or the monodromy matrix 
$M$ (see appendix~\ref{a:floquet}) associated with $A(\theta)$ is diagonalizable in $\IC$. Such properties are satisfied as soon as certain (generic) non 
resonance assumptions are made on the Lyapunov exponents thanks to the Sternberg-Chen Theorem~\cite{chen1963equivalence,nelson2015topics, WWL08}. Hence, as 
Morse-Smale flows form an open subset of all smooth vector fields, the flows we consider are in some sense generic from the Sternberg-Chen Theorem.
We refer to the appendix of~\cite{dangrivieremorsesmale1} for a detailed description of these nonresonant assumptions.

\subsection{Pollicott-Ruelle resonant states following~\cite{dangrivieremorsesmale1}}\label{ss:summary} The main goal of the 
present article is to describe both the resonances and the resonant states of $-\ml{L}_{V,\nabla}^{(k)}$ for every $0\leq k\leq n$. Recall 
from the introduction that they correspond to the poles and the residues of the meromorphic extension of the Laplace transform $\hat{C}_{\psi_1,\psi_2}$ 
of the correlation function. From~\cite{dangrivieremorsesmale1}, they are also the eigenvalues and generalized eigenmodes of the (nonseladjoint) 
operator $-\ml{L}_{V,\nabla}^{(k)}$ acting on a certain anisotropic Sobolev space of currents $\ml{H}_{k}^m(M,\ml{E})$, where 
$m(x,\xi)$ is a certain \emph{order function} which indicate the Sobolev regularity -- see section~5 of~\cite{dangrivieremorsesmale1} 
for details. More precisely, given any $T_0>0$, there exists an order function $m$ such that the spectrum of 
$-\ml{L}_{V,\nabla}^{(k)}$ on $\ml{H}_{k}^m(M,\ml{E})$ is discrete for $\text{Re}(z)\geq -T_0$. In particular, it follows 
from~\cite[Th.~1.5]{faure2011upper} that the eigenvalues and their corresponding generalized eigenmodes are independent of the choice of 
the order function $m$ satisfying the properties of Lemma~5.2 in~\cite{dangrivieremorsesmale1}. Hence, if we increase the Sobolev regularity 
in the construction (which corresponds to increase $|u|$, $|s|$ and $|n_0|$ in this Lemma), we have that a given eigenmode 
$\mathbf{u}$ stay in the anisotropic Sobolev space with the higher choice of regularity. In particular, if a generalized eigenmode $\mathbf{u}$ is supported on $W^u(\Lambda)$ 
near a critical element $\Lambda$ of the flow. Then, the wavefront set of $\mathbf{u}$~\cite{BrDaHe14a} near $\Lambda$ is contained in the 
conormal of $W^u(\Lambda)$ which roughly says that the current is smooth in the direction of $W^u(\Lambda)$.

\section{Monodromy and flat connections}\label{s:connection}

In~\cite{dang2016spectral}, we were able to compute explicitely the spectrum of $-\mathcal{L}_{V}$ by localizing the eigenvalue equation near the critical elements of the flow. Here, we will perform 
a similar analysis with the two following additional difficulties: (i) critical elements may be closed orbits of the flow 
(and not only fixed points) (ii) the complex 
vector bundle $\ml{E}$ and the corresponding flat connection $\nabla$. The risk of introducing this new geometric object $\nabla$ is that all our analysis to compute the Pollicott-Ruelle resonances
breaks down or becomes much more involved from the technical point of view. Yet, this is not the case and our goal in this section is to show that the addition of $\nabla$ does not complexify the 
calculation that much once we have defined an appropriate basis for the vector bundle $\ml{E}$.

Before stating precise results on that issue, let us start with a simple observation. Let $U$ be an open set inside $M$ and 
let $(\mathbf{c}_1,\dots,\mathbf{c}_N)$ be a moving frame of $\mathcal{E}$ defined on $U$. Then, we can write, for 
$\mathbf{u}=\sum_ju_j\mathbf{c}_j$ in $\Omega_c^\bullet(U,\mathcal{E})$
\begin{equation}\label{e:diag-action}
 \mathcal{L}_{V,\nabla}\left(\sum_{j=1}^Nu_j\mathbf{c}_j\right)=\sum_{j=1}^N\ml{L}_V(u_j)\mathbf{e}_j+\sum_{j=1}^Nu_j\nabla_V\mathbf{c}_j.
\end{equation}
Hence, if we are able to find a moving frame $(\mathbf{c}_1,\dots,\mathbf{c}_N)$ such that, for every $1\leq j\leq N$, $\nabla_V\mathbf{c}_j=\gamma_j\mathbf{c}_j$, then, for $\mathbf{u}$ in 
$\Omega_c^\bullet(U,\mathcal{E})$, one has, in some open set $\tilde{U}\subset\subset U$,
$$-\mathcal{L}_{V,\nabla}\mathbf{u}=\lambda \mathbf{u}\quad\Longleftrightarrow\quad \forall 1\leq j\leq N,\quad -\ml{L}_V(u_j)=(\lambda+\gamma_j)u_j.$$
In other words, if we have a diagonalizing moving frame, then the problem is essentially equivalent to the case of the trivial bundle $U\times\IC$. 
In this section, we shall explain how we can indeed construct such nice moving frame near any critical elements of a Morse-Smale flow. Note that the frame 
will not necessarily be as nice as above due to the fact that the monodromy matrix may not be diagonalizable.

Let us now state the precise results we shall need. We distinguish 
the cases of critical points and of closed orbits of the flow. First, in the case where $\Lambda\subset NW\left(\varphi^t\right)$ is 
a critical point, we can use the following classical result~\cite[Th.~12.25]{LeeDiff}:
\begin{theo}\label{p:good-basis-bundle} Let $\mathcal{E}\rightarrow M$ be a smooth complex vector bundle
of rank $N$ endowed with a flat connection $\nabla$. Let $U\subset M$ be a simply connected open set. 
 
Then, there exists a moving frame $(\mathbf{c}_1,\dots,\mathbf{c}_N)$ of $\mathcal{E}$ defined on $U$ such that, for every $1\leq j\leq N$, 
\begin{equation}
\nabla \mathbf{c}_j=0.
\end{equation}
\end{theo}
Thus, in the case of critical elements, we can take all the $\gamma_j$ to be equal to $0$ if we only aim at solving the eigenvalue problem locally. 
However, this is no longer the case for closed 
orbits of the flow. In fact, there may be closed orbits $\Lambda\subset NW(\varphi^t)$ for which any open neighborhood $U$ of $\Lambda$ 
is not simply connected and the above Theorem cannot 
be applied. In that case, we have to work more and we shall prove the following statement which is the main result of this section~:
\begin{theo}\label{t:good-basis-bundle} Let $\mathcal{E}\rightarrow M$ be a smooth complex vector bundle of rank $N$  
endowed with a flat connection $\nabla$. Let $V$ 
be a smooth Morse-Smale vector field on $M$ which is $\ml{C}^{\infty}$ 
linearizable and let $\Lambda$ be a closed orbit of the induced flow of minimal period $\ml{P}$.

Then there exists a neighborhood $U$ of $\Lambda$, a
smooth moving frame $(\mathbf{c}_1^{\Lambda},\dots,\mathbf{c}_N^{\Lambda})$ of $\mathcal{E}$ defined on $U$, 
some complex numbers
$(\gamma_1^{\Lambda},\dots,\gamma_N^{\Lambda})\in \mathbb{C}^N$ 
such that, for every $1\leq j\leq N$,
\begin{equation}
\nabla_V \mathbf{c}_j^{\Lambda}=\frac{2i\pi\gamma_j^{\Lambda}}{\ml{P}_{\Lambda}}\mathbf{c}_j^{\Lambda}\quad\text{or}\quad\nabla_V \mathbf{c}_j^{\Lambda}=
\frac{2i\pi\gamma_j^{\Lambda}}{\ml{P}_{\Lambda}}\mathbf{c}_j^{\Lambda}+\mathbf{c}_{j-1}^{\Lambda}.
\end{equation}
Moreover, the $e^{2i\pi\gamma_j^{\Lambda}}$ are the eigenvalues of the monodromy matrix for the parallel transport around the closed orbit $\Lambda$ and, if the 
monodromy matrix is diagonalizable, one has $$\nabla_V \mathbf{c}_j^{\Lambda}=\frac{2i\pi\gamma_j^{\Lambda}}{\ml{P}_{\Lambda}}\mathbf{c}_j^{\Lambda},$$ for every $1\leq j\leq N$.
\end{theo}
In other words, as we will be able to reduce the spectral analysis of Morse-Smale flows to solving certain eigenvalue equations near critical elements, the 
effect of introducing a flat connection will just be to shift the spectrum thanks to~\eqref{e:diag-action}. 
We will come back on that observation later on. Let us now give the proof of that Theorem.

\subsection{Trivializing $\mathcal{E}$ near a closed immersed curve $\Lambda\subset M$.} Let $\Lambda$ be a smooth closed immersed curve. We first 
prove that
the vector bundle $\mathcal{E}$ is in fact trivial near $\Lambda$:
\begin{prop}\label{p:trivial-bundle}
Let $\mathcal{E}\mapsto M$ a smooth complex vector bundle of rank $N$ on $M$. Then, near $\Lambda\subset M$, there is a neighborhood
$\mathcal{U}$ on which one has the local trivialization
$$\mathcal{E}|_{\mathcal{U}}\simeq\mathbb{S}^1\times\IR^{n-1}\times \mathbb{C}^{N}.$$
\end{prop}
We prove this proposition in two steps. First, we consider the case of an oriented real bundle on $M=\IS^1$:  
\begin{lemm}
Let $\ml{E}\rightarrow\mathbb{S}^1$ be an oriented real vector bundle of rank $n-1$ over the circle
$\mathbb{S}^1$. Then, $\ml{E}$ is trivial.
\end{lemm}
\begin{proof}
Real oriented vector bundles of rank $n-1$ over $\mathbb{S}^1$ are identified with cartesian products $[0,1]\times \mathbb{\IR}^{n-1}$ 
quotiented by an equivalence relation
$[0,1]\times \mathbb{R}^{n-1} /\sim$ identifying the fiber $\{0\}\times \mathbb{R}^{n-1}$ over $\{0\}$ with the fiber 
$\{1\}\times \mathbb{R}^{n-1}$ over $\{1\}$. In other words, there exists a 
linear invertible element $g$ in $GL_{n-1}(\mathbb{R})$ such that $(1,z)\sim (0,gz)$ and $\text{det}(g)>0$. Then, one can define a 
smooth path $g:[0,1]\mapsto GL_{n-1}(\mathbb{R})$
such that $g(0)=\text{Id}, g(1)=g$ and $\text{det}(g(\theta))>0$ for every $\theta\in[0,1]$. Let $(\mathbf{e}_1,\dots,\mathbf{e}_{n-1})$ be a basis of 
$\{0\}\times \mathbb{R}^{n-1}$ and define a moving frame
of $[0,1]\times \mathbb{R}^{n-1}$ as follows
$$ (\mathbf{e}_1(\theta),\dots,\mathbf{e}_{n-1}(\theta))=(g(\theta)^{-1}\mathbf{e}_1,\dots,g(\theta)^{-1}\mathbf{e}_{n-1}).$$
Finally, observe that over the fiber $\{1\}\times\IR^{n-1}$, $(g(1)^{-1}\mathbf{e}_1,\dots,g(1)^{-1}\mathbf{e}_{n-1})$ is identified with 
$(\mathbf{e}_1,\dots,\mathbf{e}_{n-1})=(g(0)^{-1}\mathbf{e}_1,\dots,g(0)^{-1}\mathbf{e}_{n-1})$. Hence, the moving frame 
$ (\mathbf{e}_1(\theta),\dots,\mathbf{e}_{n-1}(\theta))$ trivializes $\ml{E}$ over $\mathbb{S}^1$.  
\end{proof}
 
If $\gamma$ is a closed immersed curve inside $(M,g)$ (which is supposed to be an oriented manifold), then the above lemma shows that 
the normal bundle $N(\gamma\subset M)$ is trivial. Indeed, if one chooses an orientation of 
$\gamma$ and thus of $N(\gamma\subset M)$ (since $M$ is oriented), then $N(\gamma\subset M)$ fibers over the circle and is hence trivial. In order 
to conclude the proof of Proposition~\ref{p:trivial-bundle}, we shall now prove the following Lemma:

\begin{lemm}
Let $\ml{E}\rightarrow M$ be a smooth complex vector bundle of rank $N$. Let $\Lambda$ be a closed immersed curve in $M$, then there is a neighborhood 
$U$ of $\Lambda$ such that the bundle $\ml{E}|_U$ restricted over $U$ is \textbf{trivial}.
\end{lemm}
\begin{proof} Since the normal bundle to $\Lambda$ in $M$ is locally trivial near $\Lambda$, we can in fact consider a complex vector bundle 
$\ml{E}\rightarrow \IS^1\times\IR^{n-1}$ of rank $N$. This can be identified with the quotient cartesian product $[0,1]\times\IR^{n-1}\times\IC^N/\sim$ where, for every $x$ in $\IR^{n-1}$, there exists 
$g_x\in GL_N(\IC)$ such that  $(1,x,z)\sim (0,x,g_xz)$. Moreover, as we supposed the vector bundle to be smooth, the map $x\mapsto g_x$ can be chosen smooth. 
Then, we conclude as in the case of a real oriented bundle by observing that we can find a smooth map 
$$g:(\theta,x)\in [0,1]\times\tilde{U}\rightarrow g(\theta,x)\in Gl_N(\IC)$$ 
defined in small neighborhood of $x=0$ such that $g(0,x)=\text{Id}$ and $g(1,x)=g_x$. Indeed, one can fix a smooth function 
$\chi:[1/2,1]\rightarrow[0,1]$ which is equal to $1$ in a small neighborhood of $1/2$ and to $0$ near $1$. Thus, we can set 
$g(\theta,x)=\chi(\theta)g_0+(1-\chi(\theta))g_x.$ If $x$ is close enough to $0$, this defines a smooth path in $GL_N(\IC)$ and one 
can complete the path up to $\theta=0$ by taking a smooth path from $g_0$ to $\text{Id}$.
\end{proof}

\subsection{Parallel transport and monodromy}

Consider a path $\gamma:t\in [0,1]\mapsto \gamma(t)\in M$. An element $\mathbf{s}\in\Omega^0(M,\ml{E})$ is said to be \emph{parallel} along $\gamma$
if, it solves the following ODE:
\begin{equation}\label{e:paralleltransport}
\nabla_{\gamma^\prime(t)}\mathbf{s}(\gamma(t))=0.
\end{equation}
Given any $\mathbf{e}$ in $\ml{E}_{\gamma(0)}$, there always exists a unique parallel section along $\gamma$ satisfying 
$\mathbf{s}(\gamma(0))=\mathbf{e}$~\cite[Th.~12.20]{LeeDiff}. For $t=1$, we can then define an invertible linear mapping 
from $\ml{E}_{\gamma(0)}$ to $\ml{E}_{\gamma(1)}$, called the \textbf{holonomy} of $\gamma$ (or parallel translation):
$$H(\gamma)(\mathbf{e}):=\mathbf{s}(\gamma(1)).$$
Fix now $x\in M$ and $\gamma_x$ a loop based at $x$. The \textbf{monodromy} of the connection $\nabla$ around $\gamma_x$ calculated at $x$ reads 
\begin{equation}
M(\gamma_x)\in GL(\ml{E}_x)\simeq GL_N(\IC).
\end{equation}
When $\nabla \mathbf{s}=0$, we say that $\mathbf{s}$ is a \emph{parallel section}. Recall that such a section exists locally if we suppose that the connection 
is flat~\cite[Th.~12.25]{LeeDiff}.

\subsection{Proof of Theorem~\ref{t:good-basis-bundle}}\label{ss:proof-monodromy}

As we supposed the Morse-Smale vector field $V$ to be $\ml{C}^{\infty}$-diagonalizable, there exists some smooth local coordinates 
$(z,\theta)\in \tilde{U}\times\IR/(\ml{P}\IZ)\subset\IR^{n-1}\times\IR/(\ml{P}_{\Lambda}\IZ)$ such that the vector field has the form
$$V(z,\theta)=A_{\Lambda}(\theta)z\partial_z+\partial_{\theta}.$$
We shall now work with these local coordinates. Thanks to Proposition~\ref{p:trivial-bundle}, we can fix 
$(\mathbf{e}_1(z,\theta),\ldots, \mathbf{e}_N(z,\theta))$ to be a smooth trivializing frame field near $\Lambda:=\{(0,\theta):\theta\in\IR/(\ml{P}_{\Lambda}\IZ)\}$ 
for the complex vector $\ml{E}\rightarrow M$. In this frame field, the connection 
acts as
\begin{eqnarray*}
\nabla \left(\sum_{j=1}^Ns_j\mathbf{e}_j\right) & = & \sum_{j=1}^Nds_j\mathbf{e}_j+\sum_{j=1}^Ns_j\nabla\mathbf{e}_j\\
 & = &\sum_{j=1}^Nds_j\mathbf{e}_j+\sum_{j=1}^Ns_j\left(\sum_{l=0}^{n-1}B_l(z,\theta)dz_l+C(z,\theta) d\theta\right)\mathbf{e}_j,
\end{eqnarray*}
where $(B_l(z,\theta))_{l=1,\ldots, n-1}$ and $C(z,\theta)$ are smooth maps from $\tilde{U}\times\IR/(\ml{P}\IZ)$ to $GL_N(\IC)$. As a first step, we would like to define a new 
trivializing frame $(\mathbf{f}_1,\ldots, \mathbf{f}_N)$ such that $B\equiv 0$ and $C(z,\theta)$ does not depend on $z$. For that purpose, we fix 
$\theta_0$ in $\IR/(\ml{P}\IZ)$ and for every $1\leq j\leq N$, we solve the following system of equations:
\begin{equation}\label{e:transversal-flat}\forall 1\leq i\leq N, \nabla_{\frac{\partial}{\partial z_i}}\mathbf{f}_j(z,\theta_0)=0\ \text{and}\ \mathbf{f}_j(0,\theta_0)=\mathbf{e}_j(0,\theta_0).\end{equation}
As was already mentionned, using the fact that the connection is flat, we can find a parallel section $\tilde{\mathbf{f}}_j$ in a small 
neighborhood of $(0,\theta_0)$~\cite[Th.~12.25]{LeeDiff} satisfying $\tilde{\mathbf{f}}_j(0,\theta_0)=\mathbf{e}_j(0,\theta_0).$ Such a section solves 
in particular the above system of equations. Thus, we shall now work in this new trivializing (smooth) frame field 
$(\mathbf{f}_1,\ldots,\mathbf{f}_N)$ where, for every $1\leq j\leq N$ and for every $\theta_0\in\IR/(\ml{P}\IZ)$, $\mathbf{f}_j(z,\theta_0)$ 
is the solution of the system~\eqref{e:transversal-flat}. In this new frame, the connection reads
$$\nabla \left(\sum_{j=1}^Ns_j\mathbf{f}_j\right) =
\sum_{j=1}^Nds_j\mathbf{f}_j+\sum_{j=1}^Ns_j\left(\nabla_{\frac{\partial}{\partial\theta}}\mathbf{f}_j\right)(z,\theta) d\theta.$$
Let us now verify that $\left(\nabla_{\frac{\partial}{\partial\theta}}\mathbf{f}_j\right)(z,\theta)$ is of the form 
$\sum_{i=1}^NT(\theta)_j^i\mathbf{f}_i$ for
$z$ near $0$ where $T(\theta)$ is a smooth map from $\IR/(\ml{P}_{\Lambda}\IZ)$ to $GL_N(\IC)$. 
Indeed, by flatness of the connection, we observe that $\nabla_{\frac{\partial}{\partial\theta}}\nabla_{\frac{\partial}{\partial z}}
=\nabla_{\frac{\partial}{\partial z}}\nabla_{\frac{\partial}{\partial\theta}}$~\cite[p.~526]{LeeDiff}. Hence, using the definition of $\mathbf{f}_j$, 
one has
$$0=\nabla_{\frac{\partial}{\partial\theta}}\nabla_{\frac{\partial}{\partial z}}\mathbf{f}_j(z,\theta)
=\nabla_{\frac{\partial}{\partial z}}\nabla_{\frac{\partial}{\partial\theta}}\mathbf{f}_j(z,\theta).$$
We now decompose
$$\nabla_{\frac{\partial}{\partial\theta}}\mathbf{f}_j(z,\theta)=\sum_iT(z,\theta)_j^i\mathbf{f}_i(z,\theta),$$
and we find, using one more time the equation satisfied by $\mathbf{f}_i$,
$$0=\nabla_{\frac{\partial}{\partial z}}\left(\sum_iT(z,\theta)_j^i\mathbf{f}_i(z,\theta)\right)  =\sum_i\frac{\partial T(z,\theta)_j^i}{\partial z}\mathbf{f}_i(z,\theta).$$
Since $\mathbf{f}_i(z,\theta)_{i=1}^N$ is a basis for every $(z,\theta)$, this implies that, for every $1\leq i,j\leq N$,  
$\frac{\partial T_j^i}{\partial z}(z,\theta) =0$, hence $T(z,\theta)_j^i$ is independent of $z$ as expected.


We shall now use Floquet theory to conclude the proof of Theorem~\ref{t:good-basis-bundle}. 
From~\cite[Chapter 3]{teschl2012ordinary}, there exists 
$U(\theta,\theta_0)$ solving the following ordinary differential equation:
$$\frac{d U}{d\theta}=T(\theta)U,\quad U(\theta_0,\theta_0)=\text{Id}_{\IC^N}.$$
Moreover, this fundamental solution can be put under the form $U(\theta,0)=P_1(\theta)e^{\theta\Omega_1}$ where $P_1(\theta)$ is $\ml{P}$-periodic. We then 
set the following gauge transformation $\mathbf{f}_j^\prime=P_1(\theta)^{-1}\mathbf{f}_j$ for every $1\leq j\leq N$. In this new frame, the connection can be written as follows
$$\nabla \left(\sum_{j=1}^Ns_j\mathbf{f}_j^{\prime}\right) =
\sum_{j=1}^Nds_j\mathbf{f}_j^{\prime}+\sum_{j=1}^Ns_jT_1(\theta,d\theta)\mathbf{f}_j^{\prime} ,$$
where
$$T_1(\theta,d\theta):=P_1(\theta)^{-1}\left(-dP_1(\theta)P_1(\theta)^{-1}+T(\theta)d\theta\right)P_1(\theta)=\Omega d\theta,$$
where the second equality follows from the fact that $P_1(\theta)e^{\theta\Omega_1}$ solves the Floquet equation. Hence, in this trivializing frame, one 
has
$$\nabla \left(\sum_{j=1}^Ns_j\mathbf{f}_j^{\prime}\right) =
\sum_{j=1}^Nds_j\mathbf{f}_j^{\prime}+\sum_{j=1}^Ns_j\Omega\mathbf{f}_j^{\prime}d\theta ,$$
where $\Omega$ is a \emph{constant matrix}. Note that $e^{\ml{P}\Omega}$ is the monodromy matrix for the parallel transport along $\Lambda$. 
In order to conclude, we just use the fact that there exists a Jordan basis for the matrix $\Omega$ and make a last gauge transformation 
to work in this basis.

\subsection{Diagonalizing monodromy matrices for parallel transport}\label{ss:diagonal} Suppose now that we have a hermitian structure $\la.,.\ra_{\ml{E}}$ 
on $\ml{E}$ which preserves the flat connection, i.e. for any $\mathbf{s}_1$ and $\mathbf{s}_2$ in $\Omega^0(M,\ml{E})$, one has
$$d\left(\la \mathbf{s}_1,\mathbf{s}_2\ra_{\ml{E}}\right)=\la \nabla\mathbf{s}_1,\mathbf{s}_2\ra_{\ml{E}}+\la \mathbf{s}_1,\nabla\mathbf{s}_2\ra_{\ml{E}}.$$
Recall that complex vector bundles $\mathcal{E}\rightarrow M$ endowed with a flat connection compatible with some 
Hermitian structure are also used in Hodge theory~\cite{demailly1996intro,mnev2014lecture,bunke2015lectures}. 
In our setting, if we use the coordinates $(z,\theta)$ near the closed orbit $\Lambda$ and if we consider the parallel transport
$\mathbf{s}_1$ and $\mathbf{s}_2$ of two vectors $\mathbf{s}_1(0)$ and $\mathbf{s}_2(0)$ based at $(0,0)$
along the curve $c(\theta)=(0,\theta)$, we have
$$\frac{d}{d\theta}\left(\la \mathbf{s}_1,\mathbf{s}_2\ra_{\ml{E}}\right)=\la \nabla_{\partial_{\theta}}\mathbf{s}_1,\mathbf{s}_2\ra_{\ml{E}}
+\la \mathbf{s}_1,\nabla_{\partial_{\theta}}\mathbf{s}_2\ra_{\ml{E}}=0.$$
Hence, the parallel transport preserves the Hermitian structure on $\ml{E}$. In particular, the monodromy matrix for the parallel transport along $\Lambda$ 
is a unitary matrix and has its spectrum contained in $\IS^1=\IR/(2\pi\IZ)$. In that case, the $\gamma_j^{\Lambda}$ in Theorem~\ref{t:good-basis-bundle} 
can be chosen in $\IR$ and the effect of the flat connection is just to shift the spectrum in the vertical direction. 

In the following, we shall always assume for the sake of simplicity\footnote{At the expense of some extra combinatorial work, the general case 
could probably be treated by similar technics.} that \textbf{the Morse-Smale vector fields we consider have diagonalizable monodromy matrices for the 
 parallel transport around its closed orbits.}

\section{Reduction to a scalar problem near critical elements}\label{s:scalar}

As was already explained, we will have to solve the eigenvalue equation $-\ml{L}_{V,\nabla}^{(k)} \mathbf{u}=\lambda \mathbf{u}$ for $\mathbf{u}\in\ml{D}^{\prime k}(M,\ml{E})$ satisfying certain 
smoothness assumptions or more specifically certain wavefront assumptions -- see paragraph~\ref{s:local-form-eigenmode}. 
Theorems~\ref{p:good-basis-bundle} and~\ref{t:good-basis-bundle} provides a smooth moving frame $(\mathbf{c}_1^{\Lambda},\ldots, \mathbf{c}_N^{\Lambda})$ in the neighborhood of any 
critical element $\Lambda$ of the flow. More precisely, recall that, if we write $\mathbf{u}=\sum_{j=1}^Nu_j\mathbf{c}_j^{\Lambda}$ near $\Lambda$ (with $u_j\in\ml{D}^{\prime k}(M)$), then one has
\begin{equation}\label{e:diagonal-action-bundle}
 \ml{L}_{V,\nabla}^{(k)}\left(\sum_{j=1}^Nu_j\mathbf{c}_j^{\Lambda}\right)=\sum_{j=1}^N\left(\ml{L}_V^{(k)}+\frac{2i\pi\gamma_j^{\Lambda}}{\ml{P}_{\Lambda}}\right)u_j\mathbf{c}_j^{\Lambda},
\end{equation}
where all the $\gamma_j^{\Lambda}$ are equal to $0$ for critical points (and $\ml{P}_{\Lambda}=1$) and where 
$\left(e^{2i\pi\gamma_j^{\Lambda}}\right)_{j=1}^N$ are the eigenvalues 
of the monodromy matrix for the parallel transport around the closed orbit $\Lambda$. We would now like to reduce the problem one more time to deal only with scalar problems near the critical 
elements of the flow. For that purpose, we shall explain how to construct an appropriate basis of the vector bundle $\Lambda^k(T^*M)$ in a neighborhood of any critical element $\Lambda$ of the flow and we 
will distinguish again the case of critical points and periodic orbits. 

This construction requires a rather tedious work in order to take into account the various situations. Yet, this 
preliminary discussion permits afterwards to lighten the proofs of the upcoming sections. Along the way, 
we fix some conventions for smooth local coordinates that we shall use all along the article.

\begin{rema}
 \textbf{From this point on, we shall always suppose that $\varphi^t$ is a Morse-Smale vector field which is $\ml{C}^{\infty}$-diagonalizable}. Again, we refer to the discussion following 
 Theorem~\ref{t:maintheo-full} for comments on that assumption.
\end{rema}

\subsection{Critical points}\label{ss:local-coordinates-fixed-point} Fix $\Lambda$ a critical point of the flow. Recall that we supposed the flow to be $\ml{C}^{\infty}$-diagonalizable near $\Lambda$, i.e. there exists a 
diagonalizable (in $\IC$) matrix $A_{\Lambda}$ such that, in a smooth system of coordinates $(x_1,\ldots,x_n)$, the Morse-Smale vector field $V$ can be written as 
$$V(x,\partial_x)=A_{\Lambda}x.\partial_x.$$
Up to a linear change of coordinates in $\IR^n$, we can suppose that the matrix $A$ is block-diagonal of the form
$$A_{\Lambda}=\text{Diag}(S_1(\Lambda),\ldots, S_p(\Lambda),U_{p+1}(\Lambda),\ldots , U_q(\Lambda))\in GL_{n}(\IR),$$
where one has~:
\begin{itemize}
 \item for every $1\leq j\leq p$, one has $S_j(\Lambda)=\chi_j(\Lambda)\text{Id}_{\IR^1}$ or $\displaystyle S_{j}(\Lambda):=\left(\begin{array}{cc}\chi_j(\Lambda)& \omega_j(\Lambda)\\
-\omega_j(\Lambda) & \chi_j(\Lambda) \end{array}\right)$ with $\chi_j(\Lambda)<0$ and $\omega_j(\Lambda)\neq 0$,
  \item for every $p+1\leq j\leq q$, one has $U_j(\Lambda)=\chi_j(\Lambda)\text{Id}_{\IR^1}$ or $\displaystyle U_{j}(\Lambda):=\left(\begin{array}{cc}\chi_j(\Lambda)& \omega_j(\Lambda)\\
-\omega_j(\Lambda) & \chi_j(\Lambda) \end{array}\right)$ with $\chi_j(\Lambda)>0$ and $\omega_j(\Lambda)\neq 0$.
\end{itemize}
 The numbers $(\chi_j(\Lambda))_{j=1,\ldots, n}$ are the \textbf{Lyapunov exponents of the critical point} $\Lambda$.

We now introduce new labellings for the indices of coordinates in $\mathbb{R}^{n}$ depending on the block decomposition of the matrix $A_{\Lambda}$. We denote by 
$L_{s}$ (resp $L_u$) the set of indices $i$ corresponding to a matrix $S_i(\Lambda)$ (resp. $U_i(\Lambda)$) of size $1$ (i.e. associated with a real line) while $P_s$ (resp. $P_u$) will denote the set of indices 
corresponding to matrices of size $2$ (i.e. associated with a real plane). Then, for every $j$ in $L_s$ (resp. $L_u$), we denote the corresponding coordinates by $x_j$ (resp. $y_j$). For $P_s$ 
(resp. $P_u$), two coordinates are involved and we denote them by $(x_{j1},x_{j2})$ (resp. $(y_{j1},y_{j2})$). For latter computation, it will also be convenient to introduce the new pair
of \emph{complex} functions
\begin{eqnarray*}
z_j:(x_{j1},x_{j2})\mapsto x_{j1}+ix_{j2},\,\ \overline{z}_j: 
(x_{j1},x_{j2})\mapsto x_{j1}-ix_{j2}\\
w_j:(y_{j1},y_{j2})\mapsto y_{j1}+iy_{j2},\,\ \overline{w}_j: 
(y_{j1},y_{j2})\mapsto y_{j1}-iy_{j2}.
\end{eqnarray*}
We denote then by $(x,y)$ the two group of complex valued functions corresponding
to stable and unstable directions respectively where 
$$x=((x_j)_{j\in L_s},(z_j,\overline{z}_j)_{j\in P_s})\text{ and }y=((y_j)_{j\in L_u},(w_j,\overline{w}_j)_{j\in P_u}).$$ 
In order to make our reduction to a scalar problem, we also consider the associated Grassmann variables 
$$(dx,dy):=((dx_j)_{j\in L_s},(dz_j,d\overline{z}_j)_{j\in P_s},(dy_j)_{j\in L_u},(dw_j,d\overline{w}_j)_{j\in P_s})$$
Recall that exterior products of those anticommute. 
Keeping this convention in mind, 
we will omit to write exterior products in order to alleviate notations. Note that, for $j\in L_s$ (resp. $L_u$), one has 
$-\ml{L}_{V}^{(1)}(dx_j)=-\chi_j(\Lambda)dx_j$ (resp. $-\ml{L}_{V}^{(1)}(dy_j)=-\chi_j(\Lambda)dy_j$).  In the case where $j$ belongs to $P_s$ (resp. $P_u$), one 
has
$$ -\ml{L}_{V}^{(1)}(dz_j)=-\left(\chi_j(\Lambda)+i\omega_j(\Lambda)\right)dz_j\ \text{and}\  -\ml{L}_{V}^{(1)}(d\overline{z}_j)=-\left(\chi_j(\Lambda)-i\omega_j(\Lambda)\right)d\overline{z}_j,$$
resp.
$$ -\ml{L}_{V}^{(1)}(dw_j)=-\left(\chi_j(\Lambda)+i\omega_j(\Lambda)\right)dw_j\ \text{and}\  -\ml{L}_{V}^{(1)}(d\overline{w}_j)=-\left(\chi_j(\Lambda)-i\omega_j(\Lambda)\right)d\overline{w}_j.$$
Fix now $u(x,y,dx,dy)$ to be an element in $\ml{D}^{\prime k}(M)$ which is compactly supported in a neighborhood of $\Lambda$. We can then decompose $u$ in the 
adapted basis $(\mathbf{b}_1^{\Lambda,k},\ldots, \mathbf{b}_{N_k}^{\Lambda,k})$ given by the $k$ products of the Grassmann variables $(dx,dy)$ where $N_k:=\left(\begin{array}{c} n\\k\end{array}\right)$, i.e.
$$u(x,y,dx,dy)=\sum_{j=1}^{N_k}u_j(x,y)\mathbf{b}_j^{\Lambda,k}.$$
From our construction, we know that
$$\ml{L}_V^{(k)}(u)=\sum_{j=1}^{N_k}\left(\ml{L}_V^{(0)}(u_j)+\beta_j^{\Lambda,k}u_j\right)\mathbf{b}_j^{\Lambda,k},$$
where, for every $1\leq j\leq N_k$, one has
\begin{equation}\label{e:shift-form-critical-point}\beta_j^{\Lambda,k}=\sum_{\chi\in\text{Sp}(A_{\Lambda})}\eps_{\chi}\chi,\end{equation}
where, for every $\chi$, $\eps_{\chi}\in\{0,1\}$, $\sum_{\chi}\eps_{\chi}=k$ and where 
the eigenvalues $\chi$ are counted with their geometric multiplicity. Recall that $\chi$ is either of the form $\chi_j(\Lambda)$ or $\chi_j(\Lambda)\pm i\omega_j(\Lambda)$. Fix now 
$u\in\ml{D}^{\prime k}(M,\ml{E})$. In local coordinates near $\Lambda$, we can use the above frames to write
\begin{equation}\label{e:decompose-basis-fixed-point}u=\sum_{j=1}^N\sum_{j'=1}^{N_k} u_{jj'}\mathbf{c}_j^{\Lambda}\otimes\mathbf{b}_{j'}^{\Lambda,k},\end{equation}
with, for every $(j,j')$, $u_{jj'}\in\ml{D}^{\prime}(M)$. We can then write
\begin{equation}\label{e:lie-derivative-fixed-point}\ml{L}_{V,\nabla}^{(k)}(u)=\sum_{j=1}^N\sum_{j'=1}^{N_k} \left(\ml{L}_V(u_{jj'})+\beta_{j'}^{\Lambda,k}u_{jj'}\right)
\mathbf{c}_j^{\Lambda}\otimes\mathbf{b}_{j'}^{\Lambda,k}.\end{equation}
Hence, near a critical point, the eigenvalue equation $-\ml{L}_{V,\nabla}^{(k)}(u)=\lambda u$ can be reduced to solving some scalar eigenvalue equation where the eigenvalue $\lambda$ may be shifted.
For latter purpose, let us introduce the following set:
\begin{equation}\label{e:shift-critical-point}D_k(\Lambda):=\left\{-\beta_{j}^{\Lambda,k}:1\leq j\leq \left(\begin{array}{c} n\\k\end{array}\right)\right\}.\end{equation}
Note that every $\delta\in D_k(\Lambda)$ appears with a multiplicity which is given by
\begin{equation}\label{e:equation-multiplicity-eigenvalue-fixed-point}
 m_k(\delta,\Lambda):=N\times\left|\left\{(\eps_{\chi})_{\chi\in \text{Sp}(A)}\in\{0,1\}^n:\sum_{\chi}\eps_{\chi}=k\ \text{and}\ \delta=-\sum_{\chi}\eps_{\chi}\chi\right\}\right|.
\end{equation}
Note that, for $k=0$, $D_k(\Lambda)=\{0\}$ and $m_k(0,\Lambda)=N$.

\subsection{Periodic orbits}\label{ss:local-coordinates-closed-orbit} Fix now a periodic orbit $\Lambda$ with minimal period $\ml{P}_{\Lambda}$. We will perform a similar construction in order 
to reduce the problem to a scalar one. Recall that, as the flow is $\ml{C}^{\infty}$-linearizable, there exists a smooth system of coordinates near $\Lambda$ such that
$$V(x,\theta,\partial_x,\partial_{\theta})=A_{\Lambda}(\theta)x.\partial_x+\partial_{\theta},$$
where $A(\theta)$ is a smooth, $M_{n-1}(\IR)$-valued and $\ml{P}_{\Lambda}$-periodic function. We refer to appendix~\ref{a:floquet} for a brief reminder on Floquet theory. In particular, 
according to appendix~\ref{a:floquet} and up to a linear change of coordinates in $\IR^{n-1}$, we can suppose that the monodromy matrix $M_{\Lambda}$ is of the form
$$M_{\Lambda}=\text{Diag}(S_1(\Lambda),\ldots, S_p(\Lambda),U_{p+1}(\Lambda),\ldots , U_q(\Lambda))\in GL_{n-1}(\IR),$$
where one has
\begin{itemize}
 \item for every $1\leq j\leq p$, one has $S_j(\Lambda)=\nu_j(\Lambda)\text{Id}_{\IR^1}$  with $0<|\nu_j(\Lambda)|<1$ or $S_j(\Lambda)=\nu_j(\Lambda) R_{\vartheta_j(\Lambda)}$ with $0<\nu_j(\Lambda)<1$ and 
 $\displaystyle R_{\vartheta_j(\Lambda)}:=\left(\begin{array}{cc}\cos(\vartheta_j(\Lambda))& -\sin(\vartheta_j(\Lambda))\\
\sin(\vartheta_j(\Lambda)) & \cos(\vartheta_j(\Lambda)) \end{array}\right)$,
  \item for every $p+1\leq j\leq q$, one has $U_j(\Lambda)=\nu_j(\Lambda)\text{Id}_{\IR^1}$ with $|\nu_j(\Lambda)|>1$ or $U_j(\Lambda)=\nu_j(\Lambda) R_{\vartheta_j(\Lambda)}$ with $\nu_j(\Lambda)>1$ and 
 $\displaystyle R_{\vartheta_j(\Lambda)}:=\left(\begin{array}{cc}\cos(\vartheta_j(\Lambda))& -\sin(\vartheta_j(\Lambda))\\
\sin(\vartheta_j(\Lambda)) & \cos(\vartheta_j(\Lambda)) \end{array}\right)$.
\end{itemize}
Recall from appendix~\ref{a:floquet} that one can define the following \emph{real-valued} matrix:
$$A_{\Lambda}=\frac{1}{2\ml{P}_{\Lambda}}\log M_{\Lambda}^2.$$
As before, we introduce appropriate coordinates depending on the fact that we consider an eigenvalue associated with a line or with a two plane. One more time, we denote by 
$L_{s}$ (resp $L_u$) the set of indices $j$ corresponding to a matrix $S_i(\Lambda)$ (resp. $U_j(\Lambda)$) of size $1$ (i.e. associated with a real line) while $P_s$ (resp. $P_u$) will denote the set of indices 
corresponding to matrices of size $2$ (i.e. associated with a real plane). We can then define the \textbf{Lyapunov exponents} of the closed orbit:
$$(\chi_l(\Lambda))_{l=1}^{n-1}=\left(\left(\frac{\log|\nu_j(\Lambda)|}{\ml{P}_{\Lambda}}\right)_{j\in L_s\cup L_u},
\left(\frac{\log|\nu_j(\Lambda)|}{\ml{P}_{\Lambda}},\frac{\log|\nu_j(\Lambda)|}{\ml{P}_{\Lambda}}\right)_{j\in P_s\cup P_u}\right),$$
which are exactly the real parts of the eigenvalues of the matrix $A_{\Lambda}$. For the indices in $P_s\cup P_u$, we also set
$\omega_j(\Lambda)=\frac{\vartheta_j(\Lambda)}{\ml{P}_{\Lambda}},$ hence $\chi_j(\Lambda)\pm i\omega_j(\Lambda)$ are exactly the complex eigenvalues of $A_{\Lambda}$.

Coming back to the choice of coordinates, we also split $L_s$ (resp. $L_u$) in two disjoint parts: $L_s^+$ (resp. $L_u^+$) which correspond to some positive $\nu_j$
and $L_s^-$ (resp. $L_u^-$) which correspond to some negative $\nu_j$. Observe that $W^s(\Lambda)$ (resp. $W^u(\Lambda)$) is orientable 
if and only if $|L_s^-|$ (resp. $|L_u^{-}|$) is even. For every $j$, we define its twisting index $\tilde{\varepsilon_j}$ which is equal to $\frac{1}{2}$ 
whenever $j\in L_s^-\cup L_u^-$ and to $0$ otherwise. Again, for every $j$ in $L_s$ (resp. $L_u$), we denote the corresponding coordinates by $x_j$ (resp. $y_j$). For $P_s$ 
(resp. $P_u$), two coordinates are involved and we denote them by $(x_{j1},x_{j2})$ (resp. $(y_{j1},y_{j2})$). As for critical points, it will also be convenient to introduce the new pair
of \emph{complex} functions
\begin{eqnarray*}
z_j:(x_{j1},x_{j2})\mapsto x_{j1}+ix_{j2},\,\ \overline{z}_j: 
(x_{j1},x_{j2})\mapsto x_{j1}-ix_{j2}\\
w_j:(y_{j1},y_{j2})\mapsto y_{j1}+iy_{j2},\,\ \overline{w}_j: 
(y_{j1},y_{j2})\mapsto y_{j1}-iy_{j2}.
\end{eqnarray*}
As before, we denote by $(x,y)$ the two group of complex valued functions corresponding
to stable and unstable directions respectively where 
$$x=((x_j)_{j\in L_s},(z_j,\overline{z}_j)_{j\in P_s})\text{ and }y=((y_j)_{j\in L_u},(w_j,\overline{w}_j)_{j\in P_u}).$$ 
Then, local coordinates near $\Lambda$ are given by $(x,y,\theta)$ where $\theta$ belongs to $\IR/(\ml{P}_{\Lambda}\IZ)$. We should now define an appropriate moving frame along the closed orbit $\Lambda$. 
For that purpose, we define the associated Grassmann variables $(dx,dy,d\theta)$. According to appendix~\ref{a:floquet}, we can introduce a real valued matrix $P(\theta,0):=P(\theta)\in GL_{n-1}(\IR)$ 
which transports the stable (resp. unstable) directions along the closed orbit $\Lambda$. Using the matrix $P(\theta)^T$, we can transport the Grassmann variables along the closed 
orbit $\Lambda$ and consider all the $k$ exterior products as in the case of critical points. This gives rise to a 
moving frame for $\Lambda^k(T^*M)$ along $\Lambda$ but this frame is not a priori well defined (modulo $\ml{P}_{\Lambda}$) 
due to the fact that the matrix $P(\ml{P}_{\Lambda})$ is not equal to the identity -- see appendix~\ref{a:floquet}. 
In order to fix this problem, we just need to multiply each element in the moving frame by a 
factor $e^{2i\pi \frac{\tilde{\varepsilon}\theta}{\ml{P}_{\Lambda}}}$ where $\tilde{\varepsilon}$ is the sum of the twisting indices $\tilde{\varepsilon}_j$ 
corresponding to the Grassman variables appearing in the vector. 
We then denote this moving frame by $(b_j^{\Lambda,k}(\theta))_{j=1}^{N_k}$. As for critical points, we find that, if we write near $\Lambda$
$$u(x,y,\theta,dx,dy,d\theta)=\sum_{j=1}^{N_k}u_j(x,y,\theta)\mathbf{b}_j^{\Lambda,k},$$
then one has
$$\ml{L}_V^{(k)}(u)=\sum_{j=1}^{N_k}\left(\ml{L}_V^{(0)}(u_j)+\beta_j^{\Lambda,k}u_j\right)\mathbf{b}_j^{\Lambda,k},$$
where, for every $1\leq j\leq N_k$, one has
\begin{equation}\label{e:shift-form-closed-orbit}
\beta_j^{\Lambda,k}=\sum_{\chi\in\text{Sp}(A_{\Lambda})\cup\{0\}}\eps_{\chi}\left(\chi+\frac{2i\pi \tilde{\varepsilon}_{\chi} }{\ml{P}_{\Lambda}}\right), 
\end{equation}
where 
\begin{itemize}
 \item for every $\chi$, $\eps_{\chi}\in\{0,1\}$, $\sum_{\chi}\eps_{\chi}=k$,
 \item the eigenvalues $\chi$ are counted with their geometric multiplicity,
 \item $0$ is counted with multiplicity $1$,
 \item $\varepsilon_{\chi}$ is the twisting index.
\end{itemize}
We now fix $u$ in $\ml{D}^{\prime k}(M,\ml{E})$. In local coordinates near $\Lambda$, we can use the above frames to write
\begin{equation}\label{e:decompose-basis-closed-orbit}u=\sum_{j=1}^N\sum_{j'=1}^{N_k} u_{jj'}\mathbf{c}_j^{\Lambda}\otimes\mathbf{b}_{j'}^{\Lambda,k},\end{equation}
with, for every $(j,j')$, $u_{jj'}\in\ml{D}^{\prime}(M)$. We can then write
\begin{equation}\label{e:lie-derivative-closed-orbit}\ml{L}_{V,\nabla}^{(k)}(u)=\sum_{j=1}^N\sum_{j'=1}^{N_k} 
\left(\ml{L}_V(u_{jj'})+\left(\frac{2i\pi\gamma_j^{\Lambda}}{\ml{P}_{\Lambda}}+\beta_{j'}^{\Lambda,k}\right)
u_{jj'}\right)
\mathbf{c}_j^{\Lambda}\otimes\mathbf{b}_{j'}^{\Lambda,k}.\end{equation}
Hence, as before, the eigenvalue equation $-\ml{L}_{V,\nabla}^{(k)}(u)=\lambda u$ can be reduced to solving some scalar eigenvalue equation near $\Lambda$ (where the eigenvalue $\lambda$ may be shifted).
To mimick the conventions for critical points, let us introduce the following set:
\begin{equation}\label{e:shift-closed-orbit}D_k(\Lambda):=\left\{-(\beta_{j'}^{\Lambda,k}+\gamma_j^{\Lambda}):1\leq j\leq N,\ 1\leq j'\leq 
\left(\begin{array}{c} n\\k\end{array}\right)\right\}.\end{equation}
Again, we can compute the multiplicity of every $\delta\in D_k(\Lambda)$ with the following formula
\begin{equation}\label{e:equation-multiplicity-eigenvalue-closed-orbit}
 m_k(\delta,\Lambda):=\left|\left\{\left(j,(\eps_{\chi})_{\chi\in \text{Sp}(A_{\Lambda})\cup\{0\}}\right)
 \in\{1,\ldots,N\}\times\{0,1\}^n:(*)\ \text{holds}\right\}\right|,
\end{equation}
where $(*)$ means that
$$\sum_{\chi}\eps_{\chi}=k\ \text{and}\ \delta=-\frac{2i\pi\gamma_j^{\Lambda}}{\ml{P}_{\Lambda}}-\sum_{\chi}\eps_{\chi}\left(\chi+\frac{2i\pi\tilde{\varepsilon}_{\chi} }{\ml{P}_{\Lambda}}\right).$$

\subsection{Reduction to a scalar problem}\label{ss:reduction-scalar} Let us now summarize this construction by introducing unifying conventions for closed orbits and fixed points.
We fix $0\leq k\leq n$ and $\Lambda$ a critical element. Observe first that the sets of ``shifting parameters'' $D_k(\Lambda)$ counted 
with their multiplicities can be indexed by the set
\begin{equation}\label{e:index-shift}
 D_k:=\left\{1,\ldots, N\right\}\times\left\{1,\ldots, \left(\begin{array}{c} n\\k\end{array}\right)\right\}.
\end{equation}
In the following, we shall denote by $\mathbf{j}=(j,j')$ an element in that set and the corresponding element of $D_k(\Lambda)$ 
will be denoted by $\delta_{\mathbf{j}}^{\Lambda}$. Then, in a neighborhood of $\Lambda$, one can find a moving frame 
$(\mathbf{f}_{\mathbf{j}}^{\Lambda,k})_{\mathbf{j}\in D_{k}}$ of $\Lambda^k(T^*M)\otimes\ml{E}$ such that, 
in the system of linearized coordinates defined above, one can decompose $u\in\ml{D}^{\prime k}(M,\ml{E})$ as
$$u=\sum_{\mathbf{j}\in D_k}u_{\mathbf{j}}\mathbf{f}_{\mathbf{j}}^{\Lambda,k},$$
with $u_{\mathbf{j}}\in\ml{D}'(M)$ for every $\mathbf{j}$ in $D_k$. Moreover, one has
\begin{equation}\label{e:reduction-scalar-case}
 \ml{L}_{V,\nabla}^{(k)}\left(\sum_{\mathbf{j}\in D_k}u_{\mathbf{j}}\mathbf{f}_{\mathbf{j}}^{\Lambda,k}\right)
 =\sum_{\mathbf{j}\in D_k}\left(\ml{L}_V(u_{\mathbf{j}})-\delta_{\mathbf{j}}^{\Lambda} u_{\mathbf{j}}\right)\mathbf{f}_{\mathbf{j}}^{\Lambda,k}.
\end{equation}

\section{Solving the eigenvalue equation near critical elements}\label{s:local-form-eigenmode}

In this section, we shall explain how to solve the generalized eigenvalue equation
\begin{equation}\label{e:eigenvalue-scalar}
 \left(\ml{L}_V+\lambda\right)^p(u)=0,\ u\in\ml{D}^{\prime}(M),\ \lambda\in\IC
\end{equation}
for some $p\geq 1$
near a critical element $\Lambda$ of the Morse-Smale flow and under some constraint on the support and the regularity of $u$ near 
$\Lambda$. More precisely, we fix $\Lambda$ a critical element and $U_{\Lambda}$ a small enough\footnote{In particular, we can use the system of local coordinates introduced in 
paragraphs~\ref{ss:local-coordinates-fixed-point} and~\ref{ss:local-coordinates-closed-orbit}.} open neighborhood of $\Lambda$. We then suppose that $u$ 
solves~\eqref{e:eigenvalue-scalar} near $\Lambda$, i.e. for every $\psi\in\ml{C}^{\infty}(U_{\Lambda})$,
$$\la(\ml{L}_V+\lambda)u,\psi\ra=0.$$
We will solve this problem under the two following additional constraints:
\begin{equation}\label{e:constraint-support}
 \text{supp} (u)\cap U_{\Lambda}\subset W^u(\Lambda),
\end{equation}
and
\begin{equation}\label{e:constraint-wavefront}
 \text{WF} (u)\cap T^*U_{\Lambda}\subset \bigcup_{x\in \Lambda}N^*(W^{uu}(x)),
\end{equation}
where $W^{uu}(x)=W^{u}(\Lambda)$ in the case of a fixed point and where $W^{uu}(x)$ is defined in appendix~\ref{a:floquet} in the case of a closed orbit. These extra 
constraints are motivated by the inductive proof we will give in section~\ref{ss:inductiveproof}.

\begin{rema} 
 Recall from the description of global dynamics of Morse-Smale flows (see for instance Remark~4.5 in the proof of~\cite[Th.~4.4]{dangrivieremorsesmale1}) 
 that $U_{\Lambda}\subset W^u(\Lambda)$ is equal to the local unstable manifold for a small enough neighborhood. We refer to appendix~A 
 of~\cite{dangrivieremorsesmale1}) for a brief reminder on local unstable manifolds.
\end{rema}

As in the previous section, we will distinguish the case of critical points and the one of closed orbits in order to solve this problem. The main results of this section are 
Theorems~\ref{t:local-form-scalar-critical-point} and~\ref{t:local-form-scalar-periodic-orbit} which give the local 
form of $u$ in the system of local coordinates and the value of $\lambda$.

\subsection{Critical points}

We start with the case of critical points whose treatment is close to the one appearing for gradient flows~\cite{dang2016spectral}, up to the difference that some of the eigenvalues of $\Lambda$ may be complex. In order 
to state our main result, let us fix some conventions using the notations of paragraph~\ref{ss:local-coordinates-fixed-point}. First, to describe the value of $\lambda$, we introduce 
for every $\alpha=(\alpha_{\chi})_{\chi\in\text{Sp}(A_{\Lambda})}\in\IN^n$,
\begin{equation}\label{e:scalar-resonance-fixed-point}
\lambda_{\alpha}^{\Lambda}:=-\sum_{j\in L_s}|\chi_j(\Lambda)|-2\sum_{j\in P_s}|\chi_j(\Lambda)|+\sum_{\chi\in\text{Sp}(A_{\Lambda}):\text{Re}(\chi)<0}\alpha_{\chi}\chi-
\sum_{\chi\in\text{Sp}(A_{\Lambda}):\text{Re}(\chi)>0}\alpha_{\chi}\chi.
\end{equation}
Then, we introduce the set of resonances associated with the fixed point $\Lambda$:
\begin{equation}\label{e:set-scalar-resonance-fixed-point}\ml{R}_{0}(\Lambda):=\left\{\lambda_{\alpha}^{\Lambda}:\alpha\in\IN^n\right\}.\end{equation}
For every $\lambda\in \ml{R}_0(\Lambda)$, the multiplicity is defined as follows:
\begin{equation}\label{e:multiplicity-scalar-fixed-point}\mathfrak{m}(\lambda,\Lambda):=\left\{\alpha\in\IN^n:\lambda_{\alpha}^{\Lambda}=\lambda\right\}.\end{equation}

To these resonances, we associate some eigenmodes on $\IR^n$, i.e. for every $\alpha=(\alpha_x,\alpha_y)\in\IN^n$, we set
\begin{equation}\label{e:scalar-resonant-state-fixed-point}u_{\alpha}^{\Lambda}(x,y):=\delta_0^{(\alpha_x)}(x)y^{\alpha_y}.\end{equation}
\begin{rema}\label{r:convention-Dirac} This expression should be understood as follows. For $\alpha=(0,\ldots,0)\in\IN^r$, $\delta_0^{(0)}(x)=\delta_0(x)$ is the Dirac distribution on $\IR^r$. For every multi-index 
$\alpha:=((\alpha_j)_{j\in L_s},(\alpha_j,\overline{\alpha}_j)_{j\in P_s})\in \IN^r$, one has 
$$\delta_0^{(\alpha)}(x)=(\partial^{(\alpha_j)}_{x_j})_{j\in L_s}\left(\partial_{z_j}^{(\alpha_j)}\partial_{\overline{z}_j}^{(\overline{\alpha}_j)}\right)_{j\in P_s}(\delta_0(x)),$$
where
$$\partial_{z_j}:=\partial_{x_{j1}}+i\partial_{x_{j2}}\quad\text{and}\quad\partial_{\overline{z}_j}:=\partial_{x_{j1}}-i\partial_{x_{j2}}.$$
For the polynomial part, one has, for every $\alpha\in\IN^{n-r}$,
$$y^{\alpha}=\prod_{j\in L_u} y_j^{\alpha_j}\prod_{j\in P_u} w_j^{\alpha_j}\overline{w}_j^{\overline{\alpha}_j}.$$
\end{rema}
A more or less direct calculation shows that, using the system of local coordinates of paragraph~\ref{ss:local-coordinates-fixed-point}, one has near $\Lambda$
\begin{equation}\label{e:local-model-eigenmode-fixed-point}
 -\ml{L}_V(u_{\alpha}^{\Lambda})=\lambda_{\alpha}^{\Lambda}u_{\alpha}^{\Lambda}.
\end{equation}
These local solutions will be the building blocks to construct global resonant states for the operator $-\ml{L}_{V,\nabla}^{(k)}.$ One of their key property is the following Theorem:

\begin{theo}\label{t:local-form-scalar-critical-point} Let $\varphi^t$ be a Morse-Smale flow which is $\ml{C}^{\infty}$-diagonalizable and let $\Lambda$ be a critical point of the flow. 
Let $(u,\lambda)\in\ml{D}'(M)\times\IC$ (with $u\neq 0$) be a solution of~\eqref{e:eigenvalue-scalar} on $U_{\Lambda}$ satisfying the properties~\eqref{e:constraint-support} and~\eqref{e:constraint-wavefront}. Then, one has 
$$\lambda\in\ml{R}_{0}(\Lambda)$$
 and, in the local system of coordinates of paragraph~\ref{ss:local-coordinates-fixed-point}, there exist some constants $(c_{\alpha})_{\alpha:\lambda_{\alpha}^{\Lambda}=\lambda}$
 $$u(x,y)=\sum_{\alpha:\lambda_{\alpha}^{\Lambda}=\lambda}c_{\alpha}u_{\alpha}^{\Lambda}(x,y).$$
 
Furthermore, near the critical element $\Lambda$, $u$ solves the equation
$\left(\mathcal{L}_V+\lambda\right)u=0$. 
\end{theo}

\begin{proof} We still work in the adapted local system of coordinates $(x,y)$ near $\Lambda$ of paragraph~\ref{ss:local-coordinates-fixed-point}. We start by making use of the 
support hypothesis~\eqref{e:constraint-support} on $u$. From a classical result of Schwartz~\cite[Th.~37, p.102]{Schwartz-66}  on distributions carried by submanifolds, we know that $u$ can be written in local coordinates as
$$u(x,y)=\sum_{\alpha\in\IN^r}\delta_0^{(\alpha)}(x)u_{\alpha}(y),$$
where $u_{\alpha}$ are elements of $\ml{D}'((-\delta,\delta)^{n-r})$ (for some small $\delta$) and $u_{\alpha}=0$ is equal to $0$ except for a finite number of multi-indices. 
We then use the wavefront assumption~\eqref{e:constraint-wavefront} 
to show that the $u_{\alpha}(y)$ appearing in the local 
decomposition of $u$ are indeed smooth functions of the variable $y$~\cite[Corollary 9.3 p.~851]{DangIHP}. 
Fix now some large enough $N$ and write the Taylor expansion of every $u_{\alpha}$. One finds that
$$u=\sum_{(\alpha,\beta)\in\IN^n:|\beta|\leq C_N}c_{\alpha,\beta}\delta_0^{(\alpha)}(x)y^{\beta}+\sum_{\alpha\in\IN^r}\delta_0^{(\alpha)}(x)\ml{O}_{\alpha}(|y|^{N+1}),$$
where we recall that the sums over $\alpha$ are finite and $C_N$ is constant that depends only on $N$. Thus, for every $\psi$ compactly supported in $U_{\Lambda}$, one has
$$\la u,\psi\ra=\sum_{(\alpha,\beta)\in\IN^n:|\beta|\leq C_N}c_{\alpha,\beta}\la \delta_0^{(\alpha)}(x)y^{\beta},\psi\ra+\sum_{\alpha\in\IN^r}\la \delta_0^{(\alpha)}(x)\ml{O}_{\alpha}(|y|^{N+1}),\psi\ra.$$
Using the eigenvalue equation~\eqref{e:eigenvalue-scalar}, we find that for every test function $\psi$ supported
in the chart  of paragraph~\ref{ss:local-coordinates-fixed-point}, 
there is some polynomial $P$ such that, for every $t>0$, 
$$e^{\lambda t}P(t)=\langle \varphi^{-t*}u,\psi \rangle = \sum_{(\alpha,\beta)\in\IN^n:|\beta|\leq C_N}c_{\alpha,\beta}
e^{t\lambda_{\alpha,\beta}^{\Lambda}}\la \delta_0^{(\alpha)}(x)y^{\beta},\psi\ra+\sum_{\alpha\in\IN^r}\la \varphi^{-t*}(\delta_0^{(\alpha)}(x)\ml{O}_{\alpha}(|y|^{N+1})),\psi\ra.$$
In order to conclude, we make the Laplace transform (w.r.t. $t$) of this quantity. As $u\neq 0$, we find that, for an appropriate choice of $\psi$, the Laplace transform of the left hand-side 
is meromorphic and it has a nontrivial (a priori) multiple pole at $\lambda$. For a large enough $N$ (depending on $(u,\lambda)$), the Laplace transform of the remainder of the right hand side has no pole at $\lambda$ and the Laplace transform of the main part contains only \textbf{simple poles} at 
$\lambda_{\alpha,\beta}^{\Lambda}$ for some finite number of $(\alpha,\beta)\in\IN^n, \vert \beta\vert\leqslant C_N$.
This implies that $\lambda=\lambda_{\alpha,\beta}^{\Lambda}$ for some $(\alpha,\beta)\in\IN^n$ and that $P(t)$ is constant equal to $\la u,\psi\ra$.
Finally, using that the equality is valid for any $\psi$ and identifying the residue of the Laplace transform at $\lambda$, we obtain the expected expression for $u$ near $\Lambda$.
\end{proof}

\subsection{Periodic orbits}

We now turn to the case of closed orbits which is slightly more involved but still based on similar ideas (up to the fact that we have to work with the periodic variable $\theta$). As above, we start by 
fixing some conventions using now the notations of paragraph~\ref{ss:local-coordinates-closed-orbit}. For every $\alpha=(\alpha',\alpha_n)\in\IN^{n-1}\times\IZ$, we set
\begin{eqnarray}\label{e:scalar-resonance-periodic-orbit}
\lambda_{\alpha}^{\Lambda} &:= & -\sum_{j\in L_s}|\chi_j(\Lambda)|-2\sum_{j\in P_s}|\chi_j(\Lambda)|+i\alpha_n\frac{2\pi}{\ml{P}_{\Lambda}}\\
 & &+\sum_{\chi\in\text{Sp}(A_{\Lambda}):\text{Re}(\chi)<0}\alpha_{\chi}'\chi-
\sum_{\chi\in\text{Sp}(A_{\Lambda}):\text{Re}(\chi)>0}\alpha_{\chi}'\left(\chi+\frac{2i\pi\tilde{\varepsilon}_{\chi}}{\ml{P}_{\Lambda}}\right) ,
\end{eqnarray}
and we define
\begin{equation}\label{e:set-scalar-resonance-closed-orbit}\ml{R}_{0}(\Lambda):=\left\{\lambda_{\alpha}^{\Lambda}:\alpha\in\IN^{n-1}\times\IZ\right\}.\end{equation}
We recall that $\varepsilon_{\chi}$ denotes the twisting index of the eigenvalue. 
One more time, for every $\lambda\in \ml{R}_0(\Lambda)$, the multiplicity is defined as follows:
\begin{equation}\label{e:multiplicity-scalar-closed-orbit}\mathfrak{m}(\lambda,\Lambda):=\left\{\alpha\in\IN^{n-1}\times\IZ:\lambda_{\alpha}^{\Lambda}=\lambda\right\}.\end{equation}

We now turn to the definition of the local solutions of the problem~\eqref{e:eigenvalue-scalar}. We then set, for 
$\alpha=(\alpha_x,\alpha_y,\alpha_n)\in\IN^{n-1}\times\IZ$,
\begin{equation}\label{e:scalar-resonant-state-closed-orbit}u_{\alpha}^{\Lambda}(x,y,\theta):=
e^{\frac{2i\pi(\alpha_n-\tilde{\varepsilon}.\alpha_y) \theta}{\ml{P}_{\Lambda}}}
(P(\theta)^{-1})^*\left(\delta_0^{(\alpha_x)}(x) y^{\alpha_y}\right),\end{equation}
where we use the conventions of Remark~\ref{r:convention-Dirac} to define the derivatives of the Dirac distributions and the multinomials. Note that the prefactor involving the twisting indices 
$\varepsilon$ ensures that this is a well-defined $\ml{P}_{\Lambda}$-periodic 
distribution\footnote{Recall from appendix~\ref{a:floquet} that $P(\ml{P}_{\Lambda})$ may not be equal to the identity.}. Again, using the 
system of local coordinates of paragraph~\ref{ss:local-coordinates-fixed-point}, one has near $\Lambda$
\begin{equation}\label{e:local-model-eigenmode-closed-orbit}
 -\ml{L}_V(u_{\alpha}^{\Lambda})=\lambda_{\alpha}^{\Lambda}u_{\alpha}^{\Lambda}.
\end{equation}
\begin{rema}\label{r:propag-closed-orbit}
 In order to verify this eigenvalue equation, observe that, using appendix~\ref{a:floquet} and the conventions of
 paragraph~\ref{ss:local-coordinates-closed-orbit}, the flow can be written locally as
 $$\varphi^{-t}(x,\theta)=(P(\theta-t)e^{-tA_{\Lambda}}P(\theta)^{-1}x,\theta-t).$$
 Thus, if we write the pullback under the flow, we get
 \begin{eqnarray*}\varphi^{-t*}\left(u_{\alpha}^{\Lambda}\right) & = & e^{\frac{i\pi(-\varepsilon.\alpha_y+2\alpha_n)(\theta-t)}{\ml{P}_{\Lambda}}}
 (P(\theta-t)e^{-tA_{\Lambda}}P(\theta)^{-1})^*(P(\theta-t)^{-1})^*\left(\delta_0^{(\alpha_x)}(x) y^{\alpha_y}\right)\\
 & = & e^{\frac{i\pi(-\varepsilon.\alpha_y+2\alpha_n)(\theta-t)}{\ml{P}_{\Lambda}}}
 (P(\theta)^{-1})^*(e^{-tA_{\Lambda}})^*\left(\delta_0^{(\alpha_x)}(x) y^{\alpha_y}\right).\end{eqnarray*}  

\end{rema}

Then, the following Theorem is the analogue of Theorem~\ref{t:local-form-scalar-critical-point} for closed orbits:
\begin{theo}\label{t:local-form-scalar-periodic-orbit} Let $\varphi^t$ be a Morse-Smale flow which is $\ml{C}^{\infty}$-diagonalizable and let $\Lambda$ be a closed orbit of the flow. 
Let $(u,\lambda)\in\ml{D}'(M)\times\IC$ be a solution of~\eqref{e:eigenvalue-scalar} on $U_{\Lambda}$ (with $u\neq 0$) satisfying the properties~\eqref{e:constraint-support} and~\eqref{e:constraint-wavefront}. Then, one has 
$$\lambda\in\ml{R}_{0}(\Lambda)$$
 and, in the local system of coordinates of paragraph~\ref{ss:local-coordinates-closed-orbit}, there exist some constants $(c_{\alpha})_{\alpha:\lambda_{\alpha}^{\Lambda}=\lambda}$
 $$u(x,y,\theta)=\sum_{\alpha:\lambda_{\alpha}^{\Lambda}=\lambda}c_{\alpha}u_{\alpha}^{\Lambda}(x,y,\theta).$$
 
  Furthermore, near the critical element $\Lambda$, $u$ solves the equation
$\left(\mathcal{L}_V+\lambda\right)u=0$.
\end{theo}

\begin{proof} Recall from appendix~\ref{a:floquet} that the unstable manifold can be expressed in local coordinates as follows:
 $$W^u(\Lambda)\cap U_{\Lambda}:=\left\{(P(\theta)(0,y),\theta):y\in(-\delta,\delta)^{n-1-r},\ \theta\in\IR/(\ml{P}_{\Lambda}\IZ)\right\}.$$
 As for fixed points, we would like to use Laurent Schwartz's \cite[Theorem 37, p.102]{Schwartz-66}  
 to represent the current $u$ 
 as finite sum of distributions involving the variables $(x,y,\theta)$. More precisely, 
 according to the support assumption, the elements in the sum should be locally of the form 
 $(P(\theta)^{-1})^*(\delta_0^{(\alpha)}(x)u_{\alpha}(y,\theta)).$ However, due to the fact that $P(\ml{P}_{\Lambda})$ 
 is not equal to the identity, we can not get a global formula along the $\theta$ variables. In order to write things properly, 
 let us pass first to the universal cover $U_{\Lambda}\times\IR$ of $U_{\Lambda}\times\IR/(\ml{P}_{\Lambda}\IZ)$ by defining the 
 $2\ml{P}_{\Lambda}$-periodic distribution
 $$\tilde{u}(x,y,\theta):=P(\theta)^*u(x,y,\theta).$$
 Note that $\tilde{u}$ is not a priori $\ml{P}_{\Lambda}$-periodic in the $\theta$ variable as 
 $P(\theta)$ is only $2\ml{P}_{\Lambda}$-periodic in $\theta$.
 Now we can write thanks to Schwartz Theorem
 $$\tilde{u}(x,y,\theta)=\sum_{\alpha}\delta_0^{(\alpha)}(x)\tilde{u}_{\alpha}(y,\theta),$$
 Recall from Schwartz Theorem that the sum is finite and that $\tilde{u}_{\alpha}(y,\theta)$ is an element 
 of $\ml{D}'((-\delta,\delta)^{n-1-r}\times\IR)$ which is $2\ml{P}_{\Lambda}$-periodic in the $\theta$-variable. 
 Using the assumption~\eqref{e:constraint-wavefront} on the wave front of $u$, we know that the distributions 
 $\tilde{u}_{\alpha}(y,\theta)$ are in fact smooth by a similar argument as in~\cite[Lemma 9.2 and Corollary 9.3]{DangIHP}. As before, we
 make use of the Taylor expansion in the variable $y$, i.e. for every $\alpha\in\IN^r$ and for every $N\geq 0$, we write
 $$\tilde{u}_{\alpha}(y,\theta)=\sum_{\beta:|\beta|\leq C_N} c_{\alpha,\beta}(\theta)y^{\beta}+R_{N,\alpha}(y,\theta),$$
 with $R_{N,\alpha}(y,\theta)=\ml{O}(|y|^{N+1})$ uniformly in $\theta$ as $\tilde{u}(x,y,\theta)$ is 
 $2\ml{P}_{\Lambda}$-periodic in $\theta$. Then, we write the Fourier decomposition of each term $c_{\alpha,\beta}(\theta)$, i.e. 
$$
\tilde{u}(x,y,\theta) = \sum_{(\alpha,\beta,\tilde{\alpha}_n)\in\IN^{n-1}\times\IZ}\tilde{c}_{\alpha,\beta,\tilde{\alpha}_n}
 e^{\frac{i\pi\tilde{\alpha}_n\theta}{\ml{P}_{\Lambda}}}\delta_0^{(\alpha)}(x)y^{\beta}+\sum_{\alpha}\delta_0^{(\alpha)}(x)R_{N,\alpha}(y,\theta).
$$
Then, we can recover the expression for $u(x,y,\theta)$,
$$u(x,y,\theta) = \sum_{(\alpha,\beta,\tilde{\alpha}_n)\in\IN^{n-1}\times\IZ}\tilde{c}_{\alpha,\beta,\tilde{\alpha}_n}
 e^{\frac{i\pi\tilde{\alpha}_n\theta}{\ml{P}_{\Lambda}}}(P(\theta)^{-1})^*\left(\delta_0^{(\alpha)}(x)y^{\beta}\right)
 +\sum_{\alpha}(P(\theta)^{-1})^*\delta_0^{(\alpha)}(x)R_{N,\alpha}(y,\theta),$$
which is a $\ml{P}_{\Lambda}$ periodic distribution in $\theta$ even if each individual term is a priori only $2\ml{P}_{\Lambda}$-periodic. 
We now fix a smooth test function $\psi\in\ml{C}^{\infty}(U_{\Lambda}\times\IR)$ which is $\ml{P}_{\Lambda}$-periodic in the variable $\theta$ 
and compactly supported in $U_{\Lambda}$. We write
\begin{eqnarray*}
 \la u,\psi\ra & = &\sum_{(\alpha,\beta,\tilde{\alpha}_n)\in\IN^{n-1}\times\IZ}\tilde{c}_{\alpha,\beta,\tilde{\alpha}_n}\left\la 
 e^{\frac{i\pi\tilde{\alpha}_n\theta}{\ml{P}_{\Lambda}}}(P(\theta)^{-1})^*(\delta_0^{(\alpha)}(x)y^{\beta}),\psi\right\ra\\
 & &+\sum_{\alpha}\left\la(P(\theta)^{-1})^*(\delta_0^{(\alpha)}(x)R_{N,\alpha}(y,\theta)),\psi\right\ra.
\end{eqnarray*}
We can now make use of the eigenvalue equation~\eqref{e:eigenvalue-scalar} and of remark~\ref{r:propag-closed-orbit} 
to find that there exists some polynomial $P(t)$ such that, for every $t\geq 0$,
\begin{eqnarray*}
e^{\lambda t}P(t) &=& \la  e^{-t\mathcal{L}_V } u,\psi\ra\\
  & = &\sum_{(\alpha,\beta,\tilde{\alpha}_n)\in\IN^{n-1}\times\IZ}\tilde{c}_{\alpha,\beta,\tilde{\alpha}_n}
 e^{t\tilde{\lambda}_{\alpha,\beta,\tilde{\alpha}_n}}\left\la 
 e^{\frac{i\pi\tilde{\alpha}_n\theta}{\ml{P}_{\Lambda}}}(P(\theta)^{-1})^*(\delta_0^{(\alpha)}(x)y^{\beta}),\psi\right\ra\\
 & &+\sum_{\alpha}\left\la(\varphi^{-t*}P(\theta)^{-1})^*(\delta_0^{(\alpha)}(x)R_{N,\alpha}(y,\theta)),\psi\right\ra,
\end{eqnarray*}
where
$$
\tilde{\lambda}_{\alpha,\beta,\tilde{\alpha}_n} :=  -\sum_{j\in L_s}|\chi_j(\Lambda)|-2\sum_{j\in P_s}|\chi_j(\Lambda)|+i\tilde{\alpha}_n\frac{\pi}{\ml{P}_{\Lambda}}
+\sum_{\chi\in\text{Sp}(A_{\Lambda}):\text{Re}(\chi)<0}\alpha_{\chi}\chi-
\sum_{\chi\in\text{Sp}(A_{\Lambda}):\text{Re}(\chi)>0}\beta_{\chi}\chi.
$$
Note that this is not exactly the expression given by~\eqref{e:scalar-resonance-periodic-orbit} and we will explain how to recover the 
correct expression. Before that, we make the Laplace transform of this equality (w.r.t. the variable $t$). The left-hand side has multiple poles 
at $z=\lambda$. Using the fact that 
$R_{N,\alpha}(y,\theta)=\ml{O}(|y|^{N+1})$, we can verify that the remainder of the right-hand side has no pole at 
$z=\lambda$ provided that $N$ is chosen large enough (in a way that 
depends on $\lambda$ and $u$). Hence, arguing as in the case of fixed point, 
we find that the pole on the left hand side must be simple and $\lambda=\tilde{\lambda}_{\alpha,\beta,\tilde{\alpha}_n}$ for some $(\alpha,\beta,\tilde{\alpha}_n)$. Moreover, as the residue 
of the left hand side is equal to $\la u,\psi\ra$ near $\Lambda$, we know that $u$ is equal to
$$u(x,y,\theta)=\sum_{(\alpha,\beta,\tilde{\alpha}_n)\in\IN^{n-1}\times\IZ:\tilde{\lambda}_{\alpha,\beta,\tilde{\alpha}_n}=\lambda}
\tilde{c}_{\alpha,\beta,\tilde{\alpha}_n}
 e^{\frac{i\pi\tilde{\alpha}_n\theta}{\ml{P}_{\Lambda}}}(P(\theta)^{-1})^*(\delta_0^{(\alpha)}(x)y^{\beta}).$$
Here, we have to pay a little attention to the periodicity issue in order to conclude. First of all, observe that 
$(P(\theta)^{-1})^*(\delta_0^{(\alpha)}(x)y^{\beta})$ is not a priori $\ml{P}_{\Lambda}$-periodic due to the properties of 
$P(\theta)$ -- see appendix~\ref{a:floquet}. However, 
$e^{-\frac{2i\pi \beta.\tilde{\varepsilon} \theta}{\ml{P}_{\Lambda}}}(P(\theta)^{-1})^*(\delta_0^{(\alpha)}(x)y^{\beta})$ is $\ml{P}_{\Lambda}$-periodic. 
Thus, we can rewrite $u$ as follows:
$$u(x,y,\theta)=\sum_{(\alpha,\beta,\tilde{\alpha}_n)\in\IN^{n-1}\times\IZ:\tilde{\lambda}_{\alpha,\beta,\tilde{\alpha}_n}=\lambda}
\tilde{c}_{\alpha,\beta,\tilde{\alpha}_n}
 e^{\frac{2i\pi\left(\frac{\tilde{\alpha}_n}{2}+\beta.\tilde{\varepsilon}\right)\theta}{\ml{P}_{\Lambda}}}
 \left(e^{-\frac{i\pi \beta.\varepsilon \theta}{\ml{P}_{\Lambda}}}(P(\theta)^{-1})^*(\delta_0^{(\alpha)}(x)y^{\beta})\right).$$
Writing now that $u$ is $\ml{P}_{\Lambda}$-periodic, we finally get that
$$\sum_{(\alpha,\beta,\tilde{\alpha}_n):\tilde{\lambda}_{\alpha,\beta,\tilde{\alpha}_n}=\lambda,\ \frac{\tilde{\alpha}_n}{2}+\beta.\varepsilon\notin\IZ}
\tilde{c}_{\alpha,\beta,\tilde{\alpha}_n}
 e^{\frac{2i\pi\left(\frac{\tilde{\alpha}_n}{2}+\beta.\tilde{\varepsilon}\right)\theta}{\ml{P}_{\Lambda}}}
 \left(e^{-\frac{i\pi \beta.\tilde{\varepsilon} \theta}{\ml{P}_{\Lambda}}}(P(\theta)^{-1})^*(\delta_0^{(\alpha)}(x)y^{\beta})\right)=0,$$
 which implies that
 $$u(x,y,\theta)=\sum_{(\alpha,\beta,\tilde{\alpha}_n):\tilde{\lambda}_{\alpha,\beta,\tilde{\alpha}_n}=\lambda,\ \frac{\tilde{\alpha}_n}{2}+\beta.\varepsilon\in\IZ}
\tilde{c}_{\alpha,\beta,\tilde{\alpha}_n}
 e^{\frac{2i\pi\left(\frac{\tilde{\alpha}_n}{2}+\beta.\tilde{\varepsilon}\right)\theta}{\ml{P}_{\Lambda}}}
 \left(e^{-\frac{i\pi \beta.\tilde{\varepsilon} \theta}{\ml{P}_{\Lambda}}}(P(\theta)^{-1})^*(\delta_0^{(\alpha)}(x)y^{\beta})\right),$$
and concludes the proof. 
\end{proof}

\section{Complete description of the Pollicott-Ruelle spectrum}\label{s:proof}

In this section, we will give a complete description of the resonances of the lifted flow $\Phi_k^t$ and their corresponding resonant states. This result will make a crucial use of both the construction 
of the previous paragraphs and the anisotropic 
Sobolev spaces defined in~\cite{dangrivieremorsesmale1}. The strategy is as follows. First, we make use of the local eigenmodes 
defined in section~\ref{s:local-form-eigenmode} and we 
project each of them on a resonant state by using the spectral projectors defined in~\cite{dangrivieremorsesmale1}. 
Then, we study some of their properties~: the 
support and local form near a critical element. Finally, we show that these states generate all the resonant states of the Morse-Smale flow. This is exactly the content of Theorem~\ref{t:maintheo-full} 
which is the main result of the article and which is much more general than Theorems~\ref{mainthmintro} and~\ref{mainthmintro2} from the introduction.

\subsection{Construction of resonant states}
\label{ss:resonantstates}
 Let $\Lambda$ be a critical element of the flow (either a fixed point or a closed orbit) and let $0\leq k\leq n$. 
Fix a smooth function $0\leq\ \tau_{\Lambda}\leq 1$ which is 
equal to $1$ in a small neighborhood of $\Lambda$ and to $0$ outside a slightly larger neighborhood. 
Let $(\alpha,\mathbf{j})$ be an element in $\IN^n\times D_k$ 
(if $\Lambda$ is a fixed point) and in $\IN^{n-1}\times\IZ\times D_k$ (if $\Lambda$ is a closed orbit). Using the conventions of 
paragraph~\ref{ss:reduction-scalar} and of section~\ref{s:local-form-eigenmode}, we define the following element of 
$\ml{D}^{\prime k}(M,\ml{E})$:
\begin{equation}\label{e:resonant-state-local-form-vector}\tilde{\mathbf{u}}_{\alpha,\mathbf{j}}^{\Lambda}
 :=\tau_{\Lambda}u_{\alpha}^{\Lambda} \mathbf{f}_{\mathbf{j}}^{\Lambda, k}.
\end{equation}
In a small neighborhood of $\Lambda$, $\tilde{\mathbf{u}}_{\alpha,\mathbf{j}}^{\Lambda}$ satisfies the following equation:
\begin{equation}\label{e:eigenvalue-equation-resonant-state}-\ml{L}_{V,\nabla}^{(k)}\left(\tilde{\mathbf{u}}_{\alpha,\mathbf{j}}^{\Lambda}\right)
=\left(\lambda_{\alpha}^{\Lambda}+\delta_{\mathbf{j}}^{\Lambda}\right)\tilde{\mathbf{u}}_{\alpha,\mathbf{j}}^{\Lambda}.
\end{equation}
This element is of course not a resonant state as it only solves the eigenvalue equation locally and we 
shall now define a global resonant state from this element. For that purpose, observe that 
$\tilde{\mathbf{u}}_{\alpha,\mathbf{j}}^{\Lambda}$ belongs to the anisotropic Sobolev space $\ml{H}_k^{m}(M,\ml{E})$ 
we have constructed in~\cite[section $5$]{dangrivieremorsesmale1} provided that we choose large enough parameters for the Sobolev regularity 
in the stable and unstable directions. We can also choose the order function $m$ in a such a way that 
$\lambda_{\alpha}^{\Lambda}+\delta_{\mathbf{j}}^{\Lambda}$ belongs to the half-plane where $-\ml{L}_{V,\nabla}^{(k)}$ has 
a discrete spectrum with a finite multiplicity. In particular, if we pick a sufficiently small curve $\Gamma_{\alpha,\mathbf{j}}^{\Lambda,k}$ 
surrounding $\lambda_{\alpha}^{\Lambda}+\delta_{\mathbf{j}}^{\Lambda}$ in $\IC$, we can define the corresponding spectral projector:
\begin{equation}\label{e:spectral-projector}
 \pi_{\alpha,\mathbf{j}}^{\Lambda,k}:=\frac{1}{2i\pi}\int_{\Gamma_{\alpha,\mathbf{j}}^{\Lambda,k}}\left(z+\ml{L}_{V,\nabla}^{(k)}\right)^{-1}dz
 :\ml{H}_k^{m}(M,\ml{E})\rightarrow\ml{H}_k^{m}(M,\ml{E}).
\end{equation}
We refer to section $5$ of~\cite{dangrivieremorsesmale1} for more details on that issue. Note that this projector may a priori be equal to $0$, 
i.e. its range could be reduced to $0$. In any case, this operator allows us to define a resonant state:
\begin{equation}\label{e:global-resonant-state}
 \mathbf{u}_{\alpha,\mathbf{j}}^{\Lambda,k}:=\pi_{\alpha,\mathbf{j}}^{\Lambda,k}\left(\tilde{\mathbf{u}}_{\alpha,\mathbf{j}}^{\Lambda}\right).
\end{equation}
From the definition of the spectral projector, we know that there exists some $p\geq 1$ such that $\mathbf{u}_{\alpha,\mathbf{j}}^{\Lambda,k}$ 
solves the generalized eigenvalue equation:
\begin{equation}\label{e:generalized- eigenvalue-equation-resonant-state}
\left(\ml{L}_{V,\nabla}^{(k)}+\left(\lambda_{\alpha}^{\Lambda}+\delta_{\mathbf{j}}^{\Lambda}\right)\right)^p\mathbf{u}_{\alpha,\mathbf{j}}^{\Lambda,k}=0.
\end{equation}
We define the following family of \textbf{resonant states}
\begin{equation}\label{e:set-resonant-state}
 \ml{F}^{(k)}(V,\nabla,\ml{E}):=\ml{F}_{\text{pt}}^{(k)}(V,\nabla,\ml{E})\cup\ml{F}_{\text{orb}}^{(k)}(V,\nabla,\ml{E}),
\end{equation}
where
$$\ml{F}_{\text{pt}}^{(k)}(V,\nabla,\ml{E}):=\left\{\mathbf{u}_{\alpha,\mathbf{j}}^{\Lambda,k}:(\alpha,\mathbf{j})\in\IN^n\times D_k\right\}$$
and
$$\ml{F}_{\text{orb}}^{(k)}(V,\nabla,\ml{E}):=
\left\{\mathbf{u}_{\alpha,\mathbf{j}}^{\Lambda,k}:(\alpha,\mathbf{j})\in\IN^{n-1}\times\IZ\times D_k\right\}.$$
Theorem~\ref{t:maintheo-full} will show that the family $\ml{F}^{(k)}(V,\nabla,\ml{E})$ is composed of linearly independent vectors which generate all 
the possible resonant states. Before proving this, let us prove some properties of these resonant states (in particular 
that they are not equal to $0$):
\begin{prop}\label{p:resonant-state} Let $\Lambda$ be a critical element and let $0\leq k\leq n$. Fix $(\alpha,\mathbf{j})$ as above. Then, the 
following holds:
\begin{enumerate}
 \item $\mathbf{u}_{\alpha,\mathbf{j}}^{\Lambda,k}$ is equal to $\tilde{\mathbf{u}}_{\alpha,\mathbf{j}}^{\Lambda,k}$ in a small neighborhood of $\Lambda$;
 \item the support of $\mathbf{u}_{\alpha,\mathbf{j}}^{\Lambda,k}$ is equal to $\overline{W^u(\Lambda)}$.
\end{enumerate}
\end{prop}
\begin{rema}\label{r:linearly-independent} Note that, 
by construction, the $\tilde{\mathbf{u}}_{\alpha,\mathbf{j}}^{\Lambda,k}$ are linearly independent near every $\Lambda$. Hence, we can deduce 
from this Proposition that the elements of $\ml{F}^{(k)}(V,\nabla,\ml{E})$ are linearly independent.
\end{rema}

\begin{rema}\label{r:jordan}
From~\eqref{e:generalized- eigenvalue-equation-resonant-state}, the resonant states are not a priori solutions of a true eigenvalue 
equation and they may be associated to Jordan blocks. The works of Frenkel-Losev-Nekrasov exhibits situations where there are 
indeed infinitely many Jordan blocks~\cite{frenkel2008instantons}. 
On the other hand, one can also find simple nonresonance criteria under which there are no 
Jordan blocks~\cite[section 7]{dang2016spectral} in the case of Morse-Smale gradient flows. In any case, note that, if we restrict ourselves to the open set
$M-(\overline{W^{u}(\Lambda)}-W^{u}(\Lambda))$, this proposition combined with~\eqref{e:eigenvalue-equation-resonant-state} 
shows that $\mathbf{u}_{\alpha,\mathbf{j}}^{\Lambda,k}$ is a solution of the true eigenvalue equation on that open set.
\end{rema}

\begin{proof}
 In order to prove this Proposition, we should consider the semi-group $\Phi^{-t*}_k$ (see Remark~5.6 in~\cite{dangrivieremorsesmale1}) 
 associated to the unbounded operator
 $$-\ml{L}_{V,\nabla}^{(k)}:\ml{H}_k^{m}(M,\ml{E})\rightarrow\ml{H}_k^{m}(M,\ml{E}).$$
 In other words, we denote by $\Phi_k^{-t*}\tilde{\mathbf{u}}_{\alpha,\mathbf{j}}^{\Lambda,k}$ the solutions of the following PDE:
 $$\partial_t\mathbf{u}=-\ml{L}_{V,\nabla}^{(k)}\mathbf{u},\ \mathbf{u}(t=0)=\tilde{\mathbf{u}}_{\alpha,\mathbf{j}}^{\Lambda,k}.$$
 Then, for every $\psi\in\Omega^{n-k}(M,\ml{E}')$ and for $z$ with $\text{Re}(z)$ large enough, 
 we can define
 $$\hat{C}_{\psi}(z):=\int_0^{+\infty}\la \Phi_k^{-t*}\tilde{\mathbf{u}}_{\alpha,\mathbf{j}}^{\Lambda,k},\psi\ra e^{-zt} dt.$$ 
 Due to equation~\eqref{e:eigenvalue-equation-resonant-state}, we know that, for $\psi$ compactly supported in a small neighborhood of $\Lambda$, one 
 has $\hat{C}_{\psi}(z)=\la\tilde{\mathbf{u}}_{\alpha,\mathbf{j}}^{\Lambda,k},\psi\ra(z-(\lambda_{\alpha}^{\Lambda}+\delta_{\mathbf{j}}^{\Lambda}))^{-1}$. 
 If the support of $\psi$ does not intersect $\overline{W^u(\Lambda)}$, we can also verify that $\hat{C}_{\psi}(z)=0$.
 
 On the other hand, the main result of~\cite{dangrivieremorsesmale1} states that, for every $\psi$, the function $\hat{C}_{\psi}(z)$ has a meromorphic extension to 
 the entire complex plane. Note that the main result of the introduction in~\cite{dangrivieremorsesmale1} is stated for the trivial vector bundle 
 $M\times\IC$ and for test functions which are in $\Omega^k(M)\times\Omega^{n-k}(M)$. Yet, the results hold more generally 
 for elements $\ml{H}_m^k(M,\ml{E})\times\Omega^{n-k}(M,\ml{E})$ as it follows from the spectral analysis of the operator 
 $-\ml{L}_{V,\nabla}^{(k)}$ acting on the Hilbert space $\ml{H}_m^k(M,\ml{E})$ -- see section~$5$ of~\cite{dangrivieremorsesmale1} 
 for more details. From the construction of $\ml{H}_m^k(M,\ml{E})$, one can verify that $\tilde{\mathbf{u}}_{\alpha,\mathbf{j}}^{\Lambda,k}$ is indeed 
 in that space if we pick large enough Sobolev regularity in the definition of the order function $m$. Hence, 
 from~\cite[Sect.~5.4]{dangrivieremorsesmale1}, one knows that, near every $z_0\in \IC$, $\hat{C}_{\psi}(z)$ can be decomposed as follows:
 $$\hat{C}_{\psi}(z)=\sum_{l=1}^{m_k(z_0)}(-1)^{l-1}
 \frac{\left\la(\ml{L}_{V,\nabla}^{(k)}+z_0)^{l-1}\pi_{z_0}^{(k)}(\tilde{\mathbf{u}}_{\alpha,\mathbf{j}}^{\Lambda,k})
 ,\psi\right\ra}{(z-z_0)^l}+R_{\psi}(z),$$
 where $R_{\psi}$ is an holomorphic function and where $\pi_{z_0}^{(k)}$ is a certain spectral projector corresponding to $z_0$ 
 (and which is eventually $0$ if $z_0$ is not an eigenvalue). We can now make use of this result with 
 $z_0=-(\lambda_{\alpha}^{\Lambda}+\delta_{\mathbf{j}}^{\Lambda})$. More precisely, identifying the residue at that point when $\psi$ is compactly 
 supported near $\lambda$, we find that 
 $\la\mathbf{u}_{\alpha,\mathbf{j}}^{\Lambda,k},\psi\ra=\la\tilde{\mathbf{u}}_{\alpha,\mathbf{j}}^{\Lambda,k},\psi\ra$ for such test functions as 
 $\mathbf{u}_{\alpha,\mathbf{j}}^{\Lambda,k}$ is by definition the projection of $\tilde{\mathbf{u}}_{\alpha,\mathbf{j}}^{\Lambda,k}$. This shows the 
 first part of the proposition. For the second part, we argue similarly but with test functions $\psi$ that does intersect $\overline{W^u(\Lambda)}$. 
\end{proof}

\subsubsection{Relation to Epstein-Glaser renormalization}\label{sss:Epstein-Glaser}

In~\cite{frenkel2011instantons} and in the present paper, one 
construct some \textbf{germ of} eigenmode 
$u^\Lambda_\alpha \mathbf{f}_{\mathbf{j}}^{\Lambda,k}$ 
with corresponding
eigenvalue $\lambda$ 
locally near some critical element $\Lambda$.
Assume it is defined in some neighborhood $V_\Lambda$ of $\Lambda$.
Then
using both the flow and the fact that $u^\Lambda_\alpha \mathbf{f}_{\mathbf{j}}^{\Lambda,k}$ 
is an eigenmode for $-\mathcal{L}_{V,\nabla}^{(k)}$
allows us to define the current
$u^\Lambda_\alpha \mathbf{f}_{\mathbf{j}}^{\Lambda,k}$
in $\mathcal{D}^{\prime,k}(M\setminus \partial W^u(\Lambda),\mathcal{E})$
(where $\partial W^u(\Lambda)=\overline{ W^u(\Lambda)}\setminus  W^u(\Lambda)$ 
is the boundary of $W^u(\Lambda)$).
Indeed, for every test form $\psi\in \Omega^{n-k}(M\setminus \partial W^u(\Lambda),\mathcal{E}^\prime)$,
there is some time $T\geqslant 0$ large enough such that
$\varphi^{-T}\left(\text{supp}(\psi)\cap W^u(\Lambda)\right)$ is contained in $V_\Lambda$ and using
the relation $ e^{-T\lambda}\Phi_k^{-T*}\left(u^\Lambda_\alpha \mathbf{f}_{\mathbf{j}}^{\Lambda,k}\right)= u^\Lambda_\alpha \mathbf{f}_{\mathbf{j}}^{\Lambda,k}$, we define
the pairing $\langle u^\Lambda_\alpha \mathbf{f}_{\mathbf{j}}^{\Lambda,k},\psi \rangle$ as~:
$$ e^{-T\lambda}\langle \Phi_k^{-T*}u^\Lambda_\alpha \mathbf{f}_{\mathbf{j}}^{\Lambda,k},\psi  \rangle  .$$
So the globally defined current
$\mathbf{u}_{\alpha,\mathbf{j}}^{\Lambda,k}$ from proposition \ref{p:resonant-state}
is a \textbf{distributional extension}
of the twisted current $u^\Lambda_\alpha \mathbf{f}_{\mathbf{j}}^{\Lambda,k}\in \mathcal{D}^{\prime,k}(M\setminus \partial W^u(\Lambda),\mathcal{E})$.
This is analoguous to Epstein--Glaser renormalization~\cite{Epstein} where one aims 
to extend some distribution defined on some manifold $M$ minus some closed subset $X$ to the whole manifold 
$M$.
Therefore this motivates us to reformulate some of the results from
Theorems~\ref{t:local-form-scalar-critical-point}, \ref{t:local-form-scalar-periodic-orbit} and Proposition \ref{p:resonant-state}
as the following Theorem of independent interest~:
\begin{theo}[Extension Theorem]
\label{t:epstein-glaser}
 Let $\ml{E}\rightarrow M$ be a smooth, complex, hermitian 
vector bundle of dimension $N$ endowed with a flat unitary connection $\nabla$, 
$V$ a $\ml{C}^{\infty}$-diagonalizable Morse--Smale vector field and 
$\mathcal{L}_{V,\nabla}:\mathcal{D}^{\prime,\bullet}(M,\cE)\mapsto \mathcal{D}^{\prime,\bullet}(M,\cE)$ the corresponding linear operator acting
on twisted currents.

For every critical element
$\Lambda$
of $V$,
if
$\tilde{\mathbf{u}}$ is a germ of twisted current defined near
$\Lambda$
which is a solution of $-\mathcal{L}_{V,\nabla}\tilde{\mathbf{u}}=\lambda \tilde{\mathbf{u}}$
with the constraints
$\operatorname{supp}(\tilde{\mathbf{u}})\subset W^u(\Lambda)$ and $WF(\tilde{\mathbf{u}})\subset N^*(W^u(\Lambda))$,
then there exists a globally defined current 
$\mathbf{u}\in \mathcal{D}^{\prime}(M,\mathcal{E})$ such that
$$\left( \mathcal{L}_{V,\nabla}+z\right)^p\mathbf{u}=0 $$ for some $p\in \mathbb{N}$
and $\tilde{\mathbf{u}}=\mathbf{u}$ in some neighbborhood of $\Lambda$.
In particular $\left( \mathcal{L}_{V,\nabla}+z\right)\mathbf{u}=0$ on $M\setminus \partial W^u(\Lambda)$.
\end{theo}

\subsubsection{Relation to Ruelle-Sullivan currents}\label{sss:Ruelle-Sullivan}
In~\cite{RuSu75}, Ruelle and Sullivan constructed natural families of currents associated with Axiom A diffeomorphisms. More precisely, given a basic set 
$\Lambda$ of the diffeomorphism, they constructed two invariant De Rham currents, one $T_u$ associated with $W^u(\Lambda)$ and the other one $T_s$ 
with $W^s(\Lambda)$. The currents are \emph{not globally defined} on $M$. They are rather defined on 
$W^u(\Lambda)\cup \left(M-\overline{W^u(\Lambda)}\right)$ and their support are contained in $W^u(\Lambda)$. The difficulty for turning these currents 
into global objects is that one would have to analyse their mass near the ``boundary'' $\overline{W^u(\Lambda)}-W^u(\Lambda)$ of $W^u(\Lambda)$. In the case of diffeomorphisms derived from 
certain Morse-Smale gradient flows, this can be deduced from results of Laudenbach who analyzed the structure of $\overline{W^u(\Lambda)}$ 
as a stratified manifold~\cite{Lau92} -- see also~\cite{HaLa00, HaLa01, dang2016spectral}. From the perspective of ergodic theory, a 
crucial property of these currents is that they can be paired together and that $T_u\wedge T_s$ is an invariant probability measure 
carried by the basic set $\Lambda$ and maximizing a certain variational principle. 

Here, we deal with Morse-Smale flows which are the simplest example of Axiom A flows and whose basic sets are exactly the critical elements (either fixed 
points or closed orbits). We have constructed \emph{globally defined} 
twisted De Rham currents $\mathbf{u}^{\Lambda,k}_{\alpha,\mathbf{j}}$ which are supported by $\overline{W^u(\Lambda)}$. From their spectral definition, 
they also satisfy certain invariance relation related to the induced flow on $\ml{E}$ and each of them has an associated dual current 
$\hat{\mathbf{u}}^{\Lambda,k}_{\alpha,\mathbf{j}}$ which is supported by $\overline{W^s(\Lambda)}$ and whose local expression can be computed 
in local coordinates. Besides the fact that they are globally defined, 
we also emphasize that compared with~\cite{RuSu75}, we have infinitely many currents associated with each basic set $\Lambda$. Finally, 
for any choice of indices, we have that $\mathbf{u}^{\Lambda,k}_{\alpha,\mathbf{j}}\wedge\hat{\mathbf{u}}^{\Lambda,k}_{\alpha,\mathbf{j}}$ 
is equal to the Dirac measure at $\Lambda$ in the case of fixed points while, in the case of closed orbits, it is equal to the Lebesgue 
measure $\text{Leb}_{\Lambda}$ along $\Lambda$ with total mass 
$\ml{P}_{\Lambda}$. Hence, as for Ruelle-Sullivan currents, we recover a positive measure carried by $\Lambda$ and which is invariant under the 
flow (such a measure is unique in our case). To summarize, our \emph{infinite families of twisted currents} can be understood as generalizations of the Ruelle-Sullivan currents 
in the particular case of Morse-Smale flows satisfying proper nonresonance assumptions.

\subsection{Statement of the main Theorem}

We can now state the main result of the article which gives a complete description of the Pollicott-Ruelle spectrum:

\begin{theo}[Main Theorem]\label{t:maintheo-full} Let $V$ be a Morse-Smale vector field which is $\ml{C}^{\infty}$-diagonalizable. Let $(\ml{E},\nabla)$ be a vector bundle of rank $N$ endowed 
with a flat connection. We suppose that the Morse-Smale vector fields we consider have diagonalizable monodromy matrices for the 
 parallel transport around its closed orbits. 
 
 Then, all the elements of $\ml{F}^{(k)}(V,\nabla,\ml{E})$ are linearly independent and any resonant state of $-\ml{L}_{V,\nabla}^{(k)}$ is a linear combination of elements 
 from $\ml{F}^{(k)}(V,\nabla,\ml{E})$.
 
\end{theo}

\begin{rema}
 Recall that the first part of the Theorem is a consequence of Proposition~\ref{p:resonant-state} as was already 
 noticed in Remark~\ref{r:linearly-independent}. Hence, it remains to prove the ``generation part'' of the Theorem 
 and it will be achieved in Proposition~\ref{p:generation}.
\end{rema}

We will explain in paragraph~\ref{ss:imaginary-axis} how this Theorem implies the main results stated in the introduction. Let us 
briefly comment the assumptions we made on the vector field and the flat connection. Concerning the vector field, we already mentionned that 
Morse-Smale flows forms an open set of all smooth vector fields on $M$. Moreover, in dimension $2$, such flows are dense among all smooth vector 
fields thanks to a classical result of Peixoto~\cite{Pe62}. In higher dimension, Palis proved that they form an open set~\cite{Pa68}. Here, we require 
in addition that these flows are smoothly diagonalizable which means that 
we can find $\ml{C}^{\infty}$ charts where the flow is linearized with a block-diagonal form. Thanks to the Sternberg-Chen Theorem~\cite{chen1963equivalence,nelson2015topics,WWL08}, 
this can be guaranteed by imposing some nonresonance assumptions on the Lyapunov exponents which are generic among Morse-Smale flows. It is natural to ask if lower 
regularity assumptions would allow to conclude. The main result from~\cite{dangrivieremorsesmale1} 
was for instance valid for $\ml{C}^{1}$-linearizing assumptions, 
i.e. we have a discrete Pollicott-Ruelle spectrum on $\IC$ as soon as the vector field can be linearized in $C^1$ charts. Here, requiring $\ml{C}^{\infty}$ allows us to give a complete description 
of this spectrum. If we ask for less regularity (say $\ml{C}^k$ with $k\geq 1$), it seems to us that we could get a complete 
description of the spectrum in the half-plane $\text{Re}(z)>-C_k$ (for some $C_k>0$ depending on $k$ and on the vector field). This 
should probably be worked out by similar techniques but at the expense of a slightly more involved 
analysis (due to the low regularity) which would be beyond the scope of the present article. 
The assumption on the diagonalization of the monodromy matrix for the parallel transport around closed orbits makes also our analysis easier 
and it could probably be removed up to some extra combinatorial work. Note that our assumption is satisfied whenever 
the connection preserves an hermitian structure which is a standard hypothesis in Hodge theory~\cite{demailly1996intro}.

To prove the main result, we proceed in several steps. First, we 
start with a simple propagation Lemma which will be used at several stages of our proof 
to control supports of generalized eigencurrents.

\subsubsection{Propagation Lemmas to control supports.}

In our proof, we intend to use a simple propagation Lemma that we will now prove.
\begin{lemm}[Propagation Lemma for generalized eigenstates]\label{l:propagationzero}
Let $0\leq k\leq n$, let $z$ in $\IC$ and let $\mathbf{u}\in \mathcal{D}^{\prime,k}(M,\mathcal{E})$ be a solution 
of $\left(\ml{L}_{V,\nabla}^{(k)}+z\right)^p\mathbf{u}=0$ for some $p\geq 1$. If $\mathbf{u}|_U=0$ where $U\subset M$ is some open 
subset then $\mathbf{u}$ vanishes on the larger open subset $\bigcup_{t\in\mathbb{R}} \varphi^t(U)$. 
\end{lemm}
\begin{proof}
We choose $p$ to be the smallest integer so that
$\left(\ml{L}_{V,\nabla}^{(k)}+z\right)^p\mathbf{u}=0$. 
We shall establish the result by a duality argument. First, we note that, for every $\psi$ in $\Omega^{n-k}(M,\ml{E}')$, 
$$\frac{d}{dt}\left\la \Phi_k^{t*}\mathbf{u},\psi \right\ra=\left\la\Phi_k^{t*}\ml{L}_{V,\nabla}^{(k)} \mathbf{u},\psi\right\ra.$$

Hence a simple calculation yields~:
\begin{eqnarray*}
\frac{d}{dt}\left(\begin{array}{c}
\Phi_k^{t*}\mathbf{u}\\
\left(\ml{L}_{V,\nabla}^{(k)}+z\right)\Phi_k^{t*}\mathbf{u}\\
\dots\\
\left(\ml{L}_{V,\nabla}^{(k)}+z\right)^{p-1}\Phi_k^{t*}\mathbf{u}
\end{array} \right)=
\left(\begin{array}{cccc}
-z&1&0&\dots\\
0&-z&1&\dots\\
&&\dots &\\
0&\dots&0&-z
\end{array} \right)
\left(\begin{array}{c}
\Phi_k^{t*}\mathbf{u}\\
\left(\ml{L}_{V,\nabla}^{(k)}+z\right)\Phi_k^{t*}\mathbf{u}\\
\dots\\
\left(\ml{L}_{V,\nabla}^{(k)}+z\right)^{p-1}\Phi_k^{t*}\mathbf{u}
\end{array} \right)
\end{eqnarray*}

Hence, solving the ODE yields 
$$\Phi_k^{t*}\left(\begin{array}{c}
\mathbf{u}\\
\left(\ml{L}_{V,\nabla}^{(k)}+z\right)\mathbf{u}\\
\dots\\
\left(\ml{L}_{V,\nabla}^{(k)}+z\right)^{p-1}\mathbf{u}
\end{array} \right)=e^{-tz}e^{tN}\left(\begin{array}{c}
\mathbf{u}\\
\left(\ml{L}_{V,\nabla}^{(k)}+z\right)\mathbf{u}\\
\dots\\
\left(\ml{L}_{V,\nabla}^{(k)}+z\right)^{p-1}\mathbf{u}
\end{array} \right),\ \forall t\in\mathbb{R},$$
where
$$N:=\left(\begin{array}{cccc}
0&1&0&\dots\\
0&0&1&\dots\\
&&\dots &\\
0&\dots&0&0
\end{array} \right)$$
Choose some arbitrary element $x\in \bigcup_{t\in\mathbb{R}} \varphi^t(U)$. It means that there is some $t_0\in \mathbb{R}$ such that
$x\in \varphi^{t_0}(U)$ which is an open subset of $M$. Let $\psi$ be any vector valued test form in $\Omega^{n-k}(\varphi^{t_0}(U),\ml{E}')$. We have the identity
\begin{eqnarray*}
\langle \left(\begin{array}{c}
\mathbf{u}\\
\left(\ml{L}_{V,\nabla}^{(k)}+z\right)\mathbf{u}\\
\dots\\
\left(\ml{L}_{V,\nabla}^{(k)}+z\right)^{p-1}\mathbf{u}
\end{array} \right),\psi \rangle&=&\langle  (\Phi_k^{t_0})^*\left(\begin{array}{c}
\mathbf{u}\\
\left(\ml{L}_{V,\nabla}^{(k)}+z\right)\mathbf{u}\\
\dots\\
\left(\ml{L}_{V,\nabla}^{(k)}+z\right)^{p-1}\mathbf{u}
\end{array} \right),\Phi_{k,\dagger}^{t_0*}\psi \rangle\\
&=&\langle  e^{-t_0z}e^{tN}\left(\begin{array}{c}
\mathbf{u}\\
\left(\ml{L}_{V,\nabla}^{(k)}+z\right)\mathbf{u}\\
\dots\\
\left(\ml{L}_{V,\nabla}^{(k)}+z\right)^{p-1}\mathbf{u}
\end{array} \right),\Phi_{k,\dagger}^{t_0*}\psi \rangle=0
\end{eqnarray*}
since $\text{supp }\left(\Phi_{k,\dagger}^{t_0*}\psi\right)\subset U$ and $(\mathbf{u},(\ml{L}_{V,\nabla}^{(k)}+z)\mathbf{u},\dots
(\ml{L}_{V,\nabla}^{(k)}+z)^{p-1}\mathbf{u})|_U=0$  by definition.
Therefore $\mathbf{u}=0$ on $\varphi^{t_0}(U)$ in particular $\mathbf{u}=0$ in some neighborhood of
the given element $x\in \bigcup_{t\in\mathbb{R}} \varphi^t(U)$.
\end{proof}


\subsection{The inductive proof following the Smale quiver.}\label{ss:inductiveproof}

Following Smale's Theorem~\ref{t:smale}, we introduce an oriented graph $D$ whose $K$ vertices
are labelled by the the unstable manifolds $(W^u(\Lambda_i))_{i=1}^K$.
Recall that two vertices $W^u(\Lambda_i),W^u(\Lambda_j)$ are connected by an oriented path
starting at $W^u(\Lambda_j)$ and ending at 
$W^u(\Lambda_i)$ iff $W^u(\Lambda_j)\leqq W^u(\Lambda_i)$
i.e. $W^u(\Lambda_j)\subset\overline{W^u(\Lambda_i)}$.
This means $\Lambda_i$ is a larger stratum than
$\Lambda_j$ and $\Lambda_j$ is a stratum of the
topological boundary of $\Lambda_i$.
Let $\mathbf{u}\in \mathcal{H}^{m}_k(M,\ml{E})$
be a generalized eigencurrent which solves the equation
\begin{equation}
(\mathcal{L}^{(k)}_{V,\nabla}+z)^p\mathbf{u}=0.
\end{equation}
Our strategy is to analyze $\mathbf{u}$ near a 
\textbf{maximal critical element} $(\Lambda_j)_{j=1}^K$ 
in the Smale quiver $D$  such that $\mathbf{u}\neq 0$ near $\Lambda_j$.

\subsubsection{Eigencurrents near critical elements.}

Assume $\Lambda_j$ is a maximal element in Smale causality diagram such that $u$ does not vanish near $\Lambda_j$.
\begin{lemm}[support of generalized eigenfunctions]\label{l:support-Lemma}
Let $\mathbf{u}\in \mathcal{D}^{\prime,k}(M,\mathcal{E})$ be some
generalized eigencurrent of $-\ml{L}_{V,\nabla}^{(k)}$ acting on $\ml{H}_k^m(M)$.
If $\mathbf{u}$ vanishes in some neighborhood
of all $\Lambda_i\geqq \Lambda_j$ (with $\Lambda_i\neq\Lambda_j$) for the Smale causality relation and if the germ of distribution 
$\mathbf{u}\neq 0$ near $\Lambda_j$, then $\mathbf{u}$ is supported
on the germ of unstable manifold $W^u(\Lambda_j)$ near $\Lambda_j$.
\end{lemm}

\begin{proof}
We fix $\Lambda_j$. We use Remark 4.5 from the proof of \cite[Theorem 4.4]{dangrivieremorsesmale1} which deals 
with the construction of an invariant neighborhood of\footnote{Here and after, to alleviate notations, we sometimes note $\Lambda_i\leqq\Lambda_j$ instead of 
$W^u(\Lambda_i)\leqq W^u(\Lambda_j)$.} $\bigcup_{\Lambda_i\geqq \Lambda_j}W^s(\Lambda_i)$
for the \textbf{backward flow} $\varphi^{-t},t\geqslant 0$. Then, 
for all $\varepsilon>0$, there exists some neighborhood $\mathcal{U}$
of $\bigcup_{\Lambda_i\geqq \Lambda_j}W^u(\Lambda_i)$ which has size less than $ \varepsilon$ in the sense
that $\forall x\in \mathcal{U}, \,\
\text{dist}(x,\bigcup_{\Lambda_i\geqq \Lambda_j}W^s(\Lambda_i))\leqslant \varepsilon$ and which is invariant by the backward flow
$$\forall t\geqslant 0, \varphi^{-t}(\mathcal{U})\subset \mathcal{U}.$$

Let us consider some neighborhood $U_j$ of $\Lambda_j$ such that $U_j\subset \mathcal{U}$ and let $x\in U_j\setminus W^u(\Lambda_j)$.
Then by definition the flowline $\varphi^{-t}(x)$ cannot converge to $\Lambda_j$ when $t\rightarrow +\infty$ otherwise it would be in $W^u(\Lambda_j)$. Hence it must 
escape the small neighborhood $V_j$ of $\Lambda_j$ in finite time
and reach in the past some critical element $W^u(\Lambda_i)\geqq W^u(\Lambda_j)$ in the sense that
$\text{dist}(\varphi^{-t}(x),\Lambda_i)\rightarrow 0$ when
$t\rightarrow +\infty$.
Since the germ of vector valued current
$\left(\mathbf{u},(\mathcal{L}^{(k)}_{V,\nabla}+z)\mathbf{u},\dots,(\mathcal{L}^{(k)}_{V,\nabla}+z)^{p-1}\mathbf{u}\right)$ 
vanishes near $\Lambda_i$, we know
that 
$$\left(\mathbf{u},(\mathcal{L}^{(k)}_{V,\nabla}+z)\mathbf{u},\dots,(\mathcal{L}^{(k)}_{V,\nabla}+z)^{p-1}\mathbf{u}\right)$$ 
vanishes in some neighborhood 
$U_i$ of $\Lambda_i$ and we may choose $V_i$ small enough so that
$U_i\subset \mathcal{U}$. Let us call $T_i$ the finite
time for $\varphi^{-t}(x)$ to reach the open set $U_i$.
By continuity of $y\in U_j \mapsto \varphi^{-T_i}(y)\in \mathcal{U}$ there is some small neighborhood $U_x$ of $x$ in $U_j$ such that
$\varphi^{-T_i}(U_x)\subset U_i$ and therefore
$u$ vanishes on $\varphi^{-T_i}(U_x)$, which concludes the proof thanks to Lemma~\ref{l:propagationzero}.

\end{proof}

\subsection{Geometric structure of Jordan blocks.}

The next Lemma shows that if we have a cyclic
family $\mathbf{u}_0=\mathbf{u},\mathbf{u}_1=(\mathcal{L}_{V,\nabla}+z_0)\mathbf{u}_0,\dots,\mathbf{u}_{p-1}=(\mathcal{L}_{V,\nabla}+z_0)^{p-1}\mathbf{u}_0$  
of generalized eigencurrents generated by $\mathbf{u}$ solution
of $(\mathcal{L}_V+z_0)^p\mathbf{u}=0$
then their distributional support are ordered in some way which is organized 
by the Smale causality relation. 
In particular, it implies that 
inside a nontrivial cyclic family, generalized eigencurrents
cannot have the same distributional supports.

\begin{lemm}\label{l:jordan}[Support of currents in Jordan block]
For any element $\mathbf{u}\in \mathcal{M}^{m}_k(M,\mathcal{E})$, we 
define the subset $$\operatorname{Crit}(\mathbf{u})=\{ \Lambda \text{ s.t. the germ of current $\mathbf{u}$ does not vanish at }\Lambda  \}\subset NW(\varphi)$$ and 
$\operatorname{Max}\left( \operatorname{Crit}(\mathbf{u})\right)$ as
all maximal elements in $\operatorname{Crit}(\mathbf{u})$ for the Smale causality relation.

Assume that $\mathbf{u}$ is a generalized eigencurrent for the eigenvalue $z_0$
and $p$ is the smallest integer so that
$(\mathcal{L}^{(k)}_{V,\nabla}+z_0)^p \mathbf{u}=0$. 
Set
$\mathbf{u}_0=\mathbf{u},\mathbf{u}_1=(\mathcal{L}^{(k)}_{V,\nabla}+z_0)\mathbf{u}_0,\dots,\mathbf{u}_{p-1}=(\mathcal{L}^{(k)}_{V,\nabla}+z_0)^{p-1}\mathbf{u}_0$ 
the corresponding cyclic family.

Then
\begin{eqnarray*}
\forall \Lambda^\prime\in \operatorname{Max}(\operatorname{Crit}(\mathbf{u}_{i+1})), \exists   \Lambda\neq\Lambda'\in \operatorname{Max}(\operatorname{Crit}(\mathbf{u}_i))
\text{ s.t. }
W^u(\Lambda)\geqq W^u(\Lambda^\prime). 
\end{eqnarray*}
\end{lemm}
\begin{proof} Suppose that $p>1$ (otherwise the Lemma is empty). We obviously have the chain of inclusions
$\text{supp}(\mathbf{u}_{p-1})\subset\dots\subset \text{supp}(\mathbf{u}_0)  $ for the supports of the currents $(\mathbf{u}_0,\dots,\mathbf{u}_{p-1})$.
We shall prove the following inequality of support~:
\begin{eqnarray*}
\forall \Lambda^\prime\in \operatorname{Max}(\operatorname{Crit}(\mathbf{u}_1)), \exists   \Lambda\neq\Lambda'\in 
\operatorname{Max}(\operatorname{Crit}(\mathbf{u}_0))
\text{ s.t. }
W^u(\Lambda)\geqq W^u( \Lambda^\prime). 
\end{eqnarray*}
We proceed by contradiction and we choose some $\Lambda^\prime\in \text{Max}(\text{Crit}(\mathbf{u}_1))$ which does not satisfy the above property. 
Then by the equation $(\mathcal{L}^{(k)}_{V,\nabla}+z_0)\mathbf{u}_0=\mathbf{u}_1$, 
we find that $\Lambda^\prime\subset \text{supp }(\mathbf{u}_0)$ and, from the contradiction assumption, we deduce that
$\Lambda'\in \text{Max}(\text{Crit}(\mathbf{u}_0))$. So far, we did not use the regularity assumptions we have on the 
solutions of the eigenvalue problem. Recall that these generalized eigenmodes are intrinsic~\cite[Th.~1.5]{faure2011upper} in 
the sense that they do not depend on the choice of the order function $m$ used to define the anisotropic Sobolev spaces $\ml{H}_k^m(M,\ml{E})$ 
in~\cite[Sect.~5]{dangrivieremorsesmale1}. Hence, we can choose order functions $m$ with arbitrary high order 
of Sobolev regularity. Both $\mathbf{u}_0$ and $\mathbf{u}_1$ are supported inside $W^u(\Lambda')$ near $\Lambda'$. Hence, from 
the definition of the Sobolev space in~\cite[Sect.~5]{dangrivieremorsesmale1} and by choosing arbitrary high order of Sobolev 
regularity in these spaces, we find that the wave front of $\mathbf{u}_0$ and $\mathbf{u}_1$ are contained in the conormal of 
$W^u(\Lambda')$ near $\Lambda'$. We can now apply Theorems \ref{t:local-form-scalar-critical-point} and \ref{t:local-form-scalar-periodic-orbit}
near $\Lambda'$ combined with paragraph~\ref{ss:reduction-scalar}. We deduce that both generalized eigencurrents $(\mathbf{u}_0,\mathbf{u}_1)$ 
have local forms
\begin{eqnarray}
\mathbf{u}_0=\sum_{\lambda_\alpha^\Lambda+\delta_{\mathbf{j}}^\Lambda=z_0} c_{\Lambda,\alpha,\mathbf{j}}\tilde{\mathbf{u}}_{\alpha,\mathbf{j}}^{\Lambda,k}\\
\mathbf{u}_1=\sum_{\lambda_\alpha^\Lambda+\delta_{\mathbf{j}}^\Lambda=z_0} d_{\Lambda,\alpha,\mathbf{j}}\tilde{\mathbf{u}}_{\alpha,\mathbf{j}}^{\Lambda,k}
\end{eqnarray}
where the sum runs over multi-indices 
$\alpha$ and $\mathbf{j}\in D_k$ so that $\lambda=\lambda_\alpha^\Lambda+\delta_{\mathbf{j}}^\Lambda$ 
(recall that $\alpha\in \mathbb{N}^n$ if $\Lambda$ is a critical point
or $\alpha\in \mathbb{N}^{n-1}\times \mathbb{Z}$ if $\Lambda$ is a periodic orbit). Moreover, 
$c_{\Lambda,\alpha,\mathbf{j}},d_{\Lambda,\alpha,\mathbf{j}}$ are complex numbers
andthe coefficients $(c_{\Lambda,\alpha,\mathbf{j}})_{\lambda_\alpha^\Lambda+\delta_{\mathbf{j}}^\Lambda=z_0}$
do not all vanish since the germ $\mathbf{u}_0\neq 0$ near $\Lambda$.
But from the definition of the currents $\mathbf{u}_{\alpha,\mathbf{j}}^{\Lambda,k}$, we find that locally near $\Lambda$,
$(\mathcal{L}_{V,\nabla}^{(k)}+z_0)\mathbf{u}_0=0 $ hence $\mathbf{u}_1=(\mathcal{L}^{(k)}_{V,\nabla}+z_0)\mathbf{u}_0=0$ near $\Lambda$ which contradicts 
the fact that the germ $\mathbf{u}_1$ does not vanish near $\Lambda$. This gives the expected contradiction. Repeating the argument, we can deduce the 
Lemma by induction over $i$.
\end{proof}

Finally the next proposition concludes the proof of Theorem~\ref{t:maintheo-full}~:
\begin{prop}[Generation]\label{p:generation}
Assume that $\mathbf{u}\in\ml{H}_k^m(M,\ml{E})$ 
is a generalized eigencurrent for the eigenvalue $z_0$.
Then $\mathbf{u}$ is a linear combination of the eigencurrents
$\mathbf{u}_{\alpha,\mathbf{j}}^{\Lambda,k}$ for all triples
$(\Lambda,\alpha,\mathbf{j})$ such that
$\lambda_\alpha^\Lambda+\delta_{\mathbf{j}}^\Lambda=z_0$.
\end{prop}
\begin{proof}
In the proof of Lemma~\ref{l:jordan}, we found that
$\mathbf{u}=\sum_{\lambda_\alpha^\Lambda+\delta_{\mathbf{j}}^\Lambda=z_0} c_{\Lambda,\alpha,\mathbf{j}}\mathbf{u}_{\alpha,\mathbf{j}}^{\Lambda,k} $
near all elements $\Lambda\in \text{Max}(\text{Crit}(\mathbf{u}_0))$.
Therefore we define the current
$$\mathbf{v}_1=\mathbf{u}-\sum_{\Lambda\in \text{Max}(\text{Crit}(\mathbf{u}_0))} 
\sum_{\lambda_\alpha^\Lambda+\delta_{\mathbf{j}}^\Lambda=z_0} c_{\Lambda,\alpha,\mathbf{j}}\mathbf{u}_{\alpha,\mathbf{j}}^{\Lambda,k}$$ 
which is a generalized eigencurrent for the eigenvalue
$z_0$ which vanishes near all critical elements
$\Lambda\in \text{Max}(\text{Crit}(\mathbf{u}_0))$. 

Now repeat the algorithm for the new generalized eigencurrent $\mathbf{v}_1$ whose support
is represented by lower vertices in the Smale quiver thanks to the propagation Lemma \ref{l:propagationzero}.
Again, we find that
$\mathbf{v}_1=\sum_{\lambda_\alpha^\Lambda+\delta_{\mathbf{j}}^\Lambda=z_0} c_{\Lambda,\alpha,\mathbf{j}}\mathbf{u}_{\alpha,\mathbf{j}}^{\Lambda,k} $
near all elements $\Lambda\in \text{Max}(\text{Crit}(\mathbf{v}_1))$.
Therefore the current
$$\mathbf{v}_2:=\mathbf{v}_1-\sum_{\Lambda\in \text{Max}(\text{Crit}(\mathbf{v}_1))} 
\sum_{\lambda_\alpha^\Lambda+\delta_{\mathbf{j}}^\Lambda=z_0} c_{\Lambda,\alpha,\mathbf{j}}\mathbf{u}_{\alpha,\mathbf{j}}^{\Lambda,k}$$ is a generalized eigencurrent for the eigenvalue
$z_0$ vanishing near all critical elements
$\Lambda\in \text{Max}(\text{Crit}(\mathbf{v}_1))$. 
Since the number of critical elements is finite, the algorithm starting from $\mathbf{u}_0=\mathbf{v}_0$
will terminate at some element $\mathbf{v}_p$ once we exhausted all critical elements and we find
$$\mathbf{u}=\sum_{l=0}^p  \sum_{\Lambda\in \text{Max}(\text{Crit}(\mathbf{v}_l))} \sum_{\lambda_\alpha^\Lambda+\delta_{\mathbf{j}}^\Lambda=z_0} c_{\Lambda,\alpha,\mathbf{j}}\mathbf{u}_{\alpha,\mathbf{j}}^{\Lambda,k} $$
and the decomposition is unique from Remark~\ref{r:linearly-independent}.
\end{proof}

\subsection{Connection preserving a hermitian structure and spectrum on the imaginary axis}\label{ss:imaginary-axis}

Let us now explain how we can derive Theorem~\ref{mainthmintro} from the introduction (the proof of Theorem~\ref{mainthmintro2} 
follows similar lines). Recall that we only stated there Theorems valid under the 
simplifying assumption that $\nabla$ preserves an Hermitian structure on $\ml{E}$ and that we only described the Pollicott-Ruelle spectrum 
on the imaginary axis $\text{Re}(z)=0$. In that case, we recall from paragraph~\ref{ss:diagonal} that the eigenvalues of the 
monodromy matrix for the parallel transport lies on the unit circle. More specifically, the eigenvalues of the monodromy matrix are 
of the form $\rho_j^{\Lambda}(\ml{E})=e^{2i\pi\gamma_j^{\Lambda}}$ for $1\leq j\leq N$ with $\gamma_j^{\Lambda}\in\IR$ the quantity involved in the 
definition of the shifting parameter $\delta_{\mathbf{j}}^{\Lambda}$ -- see equations~\eqref{e:shift-closed-orbit} 
and~\eqref{e:generalized- eigenvalue-equation-resonant-state}.

Let us now describe the eigenvalues of $-\ml{L}_{V,\nabla}^{(k)}$ satisfying $\text{Re}(z)=0$. Recall that eigenvalues are of the form 
$\lambda_{\alpha}^{\Lambda}+\delta_{\mathbf{j}}^{\Lambda}$ where 
\begin{itemize}
 \item for critical points, $\lambda_{\alpha}^{\Lambda}$ is defined 
by~\eqref{e:scalar-resonance-fixed-point} and $\delta_{\mathbf{j}}^{\Lambda}=\beta_{j'}^{\Lambda,k}$ for every 
$\mathbf{j}=(j,j')\in D_k=\{1,\ldots,N\}\times\{1,\ldots, n!/(k!(n-k)!)\}$ with $\beta_{j'}^{\Lambda,k}$ defined 
by~\eqref{e:shift-form-critical-point};
 \item for closed orbits, $\lambda_{\alpha}^{\Lambda}$ is defined 
by~\eqref{e:scalar-resonance-periodic-orbit} and $\delta_{\mathbf{j}}^{\Lambda}=\beta_{j'}^{\Lambda,k}+\frac{2i\pi \gamma_j^{\Lambda}}{\ml{P}_{\Lambda}}$ 
for every $\mathbf{j}=(j,j')\in D_k=\{1,\ldots,N\}\times\{1,\ldots, n!/(k!(n-k)!)\}$ with $\beta_{j'}^{\Lambda,k}$ defined 
by~\eqref{e:shift-form-closed-orbit}.
\end{itemize}

Recall that we defined $\varepsilon_{\Lambda}$ to be equal to $0$ whenever $W^u(\Lambda)$ is orientable and to $1/2$ otherwise. Combining these observations 
with the local forms of the resonant states near $\Lambda$ (given by~\eqref{e:local-model-eigenmode-fixed-point} 
and~\eqref{e:local-model-eigenmode-closed-orbit}), we find, for every $0\leq k\leq n$ the following resonances on the imaginary axis:
\begin{itemize}
 \item For every \textbf{fixed point $\Lambda$ such that $\text{dim}\ W^s(\Lambda)=k$}, we have a resonance of multiplicity $N$ at 
 $$z=0,$$ 
 each of them being associated with a generalized eigenmode whose local form near $\Lambda$ is given by
 $$\forall 1\leq j\leq N,\ \delta_0(x)\left(\wedge_{j\in L_s} dx_j \right)\wedge \left(\wedge_{j\in P_s} (dz_j\wedge d\overline{z}_j)\right)\mathbf{c}_j^{\Lambda},$$
 with the conventions of paragraph~\ref{ss:local-coordinates-fixed-point} for the local coordinates and with $\nabla\mathbf{c}_j^{\Lambda}=0$ 
 for every $j$ -- see Theorem~\ref{p:good-basis-bundle}.
 \item For every \textbf{closed orbit $\Lambda$ such that $\text{dim}\ W^s(\Lambda)=k+1$}, we have infinitely many resonances on the imaginary 
 axis. More precisely, for every $1\leq j\leq N$ and for every $m\in\IZ$, we have a 
 resonance which is given by $$z=-\frac{2i\pi}{\ml{P}_{\Lambda}}\left(m+\varepsilon_{\Lambda}+\gamma_j^{\Lambda}\right)$$ 
 associated with a resonant state whose local form near $\Lambda$ is 
 given by
$$e^{\frac{2i\pi \theta}{\ml{P}_{\Lambda}}(m+\varepsilon_{\Lambda})}(P(\theta)^{-1})^*\left(\delta_0(x)\left(\wedge_{j\in L_s} dx_j \right)\wedge \left(\wedge_{j\in P_s} (dz_j\wedge d\overline{z}_j)\right)\right)
\mathbf{c}_j^{\Lambda},$$
with the conventions of paragraph~\ref{ss:local-coordinates-fixed-point} for the local coordinates and with 
$\nabla\mathbf{c}_j^{\Lambda}=\frac{2i\pi\gamma_j^\Lambda}{\ml{P}_{\Lambda}}\mathbf{c}_j^{\Lambda}d\theta$ 
 for every $j$ -- see Theorem~\ref{t:good-basis-bundle} and paragraph~\ref{ss:proof-monodromy}. 
\item For every \textbf{closed orbit $\Lambda$ such that $\text{dim}\ W^s(\Lambda)=k$}, we have infinitely many resonances on the imaginary 
 axis. More precisely, for every $1\leq j\leq N$ and for every $m\in\IZ$, we have a 
 resonance which is given by $$z=-\frac{2i\pi}{\ml{P}_{\Lambda}}\left(m+\varepsilon_{\Lambda}+\gamma_j^{\Lambda}\right)$$ 
 associated with a resonant state whose local form near $\Lambda$ is 
 given by
$$e^{\frac{2i\pi \theta}{\ml{P}_{\Lambda}}(m+\varepsilon_{\Lambda})}(P(\theta)^{-1})^*\left(\delta_0(x)\left(\wedge_{j\in L_s} dx_j \right)\wedge \left(\wedge_{j\in P_s} (dz_j\wedge d\overline{z}_j)\right)\right)
\wedge d\theta
\mathbf{c}_j^{\Lambda},$$
with the conventions of paragraph~\ref{ss:local-coordinates-fixed-point} for the local coordinates and with 
$\nabla\mathbf{c}_j^{\Lambda}=\frac{2i\pi\gamma_j^\Lambda}{\ml{P}_{\Lambda}}\mathbf{c}_j^{\Lambda}d\theta$  
 for every $j$ -- see Theorem~\ref{t:good-basis-bundle} and paragraph~\ref{ss:proof-monodromy}.
\end{itemize}

\subsection{Weyl's asymptotics}\label{ss:weyl}

As an application of our main Theorem~\ref{t:maintheo-full}, we can give asymptotics formulas for the resonance couting function. We will use the conventions of 
paragraphs~\ref{ss:local-coordinates-fixed-point} and~\ref{ss:local-coordinates-closed-orbit}. In order to alleviate the expressions, we introduce the following conventions:
$$\chi^+_{\Lambda}:=\left(|\text{Re}(\chi)|\right)_{\chi\in\text{Sp}(A_{\Lambda})}\ \text{and}\ \omega_{\Lambda}:=\left(\text{Im}(\chi)\right)_{\chi\in\text{Sp}(A_{\Lambda})},$$
where the eigenvalues are indexed with their algebraic multiplicity which is equal to the geometric multiplicity (as we supposed $A_{\Lambda}$ to be diagonalizable in $\IC$). 
In the case of a fixed point, these define vectors in $\IR^n$ and, in the case of closed orbits, 
vectors in $\IR^{n-1}$. Then, we define a \emph{convex polytope associated to every critical element} $\Lambda$. More precisely, for a fixed point $\Lambda$, we set
\begin{equation}\label{e:polytope-fixed-point}\ml{Q}_{\Lambda}:=\left\{x\in\IR^n_+:-1\leq x.\omega_{\Lambda}\leq 1\ \text{and}\ x.\chi^+_{\Lambda}\leq 1
\right\},\end{equation}
and, for a closed orbit $\Lambda$, 
\begin{equation}\label{e:polytope-closed-orbit}\ml{Q}_{\Lambda}:=\left\{(x',x_n)\in\IR^{n-1}_+\times\IR:-1\leq x'.\omega_{\Lambda}+\frac{2\pi}{\ml{P}_{\Lambda}}x_n\leq 1\ \text{and}\ x'.\chi^+_{\Lambda}\leq 1\right\}.\end{equation}
With these conventions, one has 
\begin{coro}[Weyl's law] Suppose that the assumptions of Theorem~\ref{t:maintheo-full} are satisfied. Then, for every $0\leq k\leq n$, one has, as $T\rightarrow+\infty$,
\begin{eqnarray*} N_k(T) &:=&\left|\left\{z_0\in\ml{R}_k(V,\nabla):\ |\operatorname{Im}(z_0)|\leq T\ \text{and}\ \operatorname{Re}(z_0)\geq -T\right\}\right|\\
&=& \frac{N n!}{k! (n-k)!}\left(\sum_{\Lambda}\operatorname{Vol}_{\IR^n}(\ml{Q}_{\Lambda})\right)T^n+\ml{O}(T^{n-1}).
\end{eqnarray*}
where the Pollicott-Ruelle resonances are counted with their algebraic multiplicity.
\end{coro}
\begin{proof} Let $0\leq k\leq n$ and let $T>0$ be some large enough paramater. First of all, observe from Theorem~\ref{t:maintheo-full} 
combined with~\eqref{e:set-resonant-state} and~\eqref{e:generalized- eigenvalue-equation-resonant-state} that
\begin{eqnarray*}N_k(T)  & = 
&\sum_{\Lambda\ \text{fixed point}} \left|\left\{(\alpha,\mathbf{j})\in\IN^n\times D_k:\ |\operatorname{Im}(\lambda_{\alpha}^{\Lambda}+\delta_{\mathbf{j}}^{\Lambda})|\leq T\ \text{and}\ \operatorname{Re}(\lambda_{\alpha}^{\Lambda}+\delta_{\mathbf{j}}^{\Lambda})\geq -T
\right\}\right|\\
  & +&\sum_{\Lambda\ \text{closed orbit}}\left|\left\{(\alpha,\mathbf{j})\in\IN^{n-1}\times\IZ\times D_k:\ |\operatorname{Im}(\lambda_{\alpha}^{\Lambda}+\delta_{\mathbf{j}}^{\Lambda})|\leq T\ \text{and}\ \operatorname{Re}(\lambda_{\alpha}^{\Lambda}+\delta_{\mathbf{j}}^{\Lambda})\geq -T
\right\}\right|.
\end{eqnarray*}
As we are only interested as the behaviour as $T\rightarrow+\infty$ and as $|D_k|=\frac{N n!}{k! (n-k)!}$, this quantity is in fact equal to
\begin{equation}\label{e:weyl-split}N_k(T)=\frac{N n!}{k! (n-k)!} \sum_{\Lambda} N^{\Lambda}_0(T)+\ml{O}(T^{n-1}),\end{equation}
where, for every fixed point $\Lambda$,
$$N^{\Lambda}_0(T)=\left|\left\{\alpha\in\IN^n:\ |\operatorname{Im}(\lambda_{\alpha}^{\Lambda})|\leq T\ \text{and}\ \operatorname{Re}(\lambda_{\alpha}^{\Lambda})\geq -T
\right\}\right|$$
and, for every closed orbit $\Lambda$,
$$N^{\Lambda}_0(T)=\left|\left\{\alpha\in\IN^{n-1}\times\IZ:\ |\operatorname{Im}(\lambda_{\alpha}^{\Lambda})|\leq T\ \text{and}\ \operatorname{Re}(\lambda_{\alpha}^{\Lambda})\geq -T
\right\}\right|.$$
Let us now compute the asymptotic formula for $N_0^{\Lambda}(T)$ for every choice of $\Lambda$. For a fixed point $\Lambda$, we find, as $T\rightarrow+\infty$,
$$N^{\Lambda}_0(T)=\left|\left\{\alpha\in\IN^n:\ -T\leq\alpha.\omega_{\Lambda}\leq T\ \text{and}\ 0\leq \alpha.\chi^+_{\Lambda}\leq T
\right\}\right|+\ml{O}(T^{n-1}).$$
Up to an error term of order $\ml{O}(T^{n-1})$, this can be rewritten under an integral form, i.e.
$$N^{\Lambda}_0(T)=\text{Vol}_{\IR^n}\left(\left\{x\in\IR^n_+:-T\leq x.\omega_{\Lambda}\leq T\ \text{and}\ x.\chi^+_{\Lambda}\leq T
\right\}\right)+\ml{O}(T^{n-1}),$$
or equivalently
$$N^{\Lambda}_0(T)=T^n\text{Vol}_{\IR^n}\left(\left\{x\in\IR^n_+:-1\leq x.\omega_{\Lambda}\leq 1\ \text{and}\ x.\chi^+_{\Lambda}\leq 1
\right\}\right)+\ml{O}(T^{n-1}).$$
For a closed orbit, the same calculation would give
$$N^{\Lambda}_0(T)=T^n\text{Vol}_{\IR^n}\left(\left\{(x',x_n)\in\IR^{n-1}_+\times\IR:-1\leq x'.\omega_{\Lambda}+\frac{2\pi}{\ml{P}_{\Lambda}}x_n\leq 1\ \text{and}\ x'.\chi^+_{\Lambda}\leq 1\right\}\right)+\ml{O}(T^{n-1}),$$
which, combined with~\eqref{e:weyl-split}, concludes the proof of the corollary.
\end{proof}

\appendix

\section{Monodromy and periodic orbits: a brief account of Floquet theory}\label{a:floquet}

In section~\ref{s:morse-smale}, we made the assumption that near every closed orbit $\Lambda$ of the flow, 
one can choose (regular) local coordinates such that the vector field $V$ 
can be put into the normal form
$$V(x,\theta):=A(\theta)x.\partial_x+\partial_{\theta},$$
where $x\in\IR^{n-1}$, $\theta\in\IR/(\ml{P}_{\Lambda}\IZ)$, and $A(\theta):\IR/(\ml{P}_{\Lambda}\IZ)\rightarrow GL_{n-1}(\IR)$. 
Floquet theory describes the solution of the 
following ordinary differential equation:
\begin{equation}\label{e:Floquet}
\frac{dU}{d\theta}=A(\theta)U,\quad U(\theta_0,\theta_0)=\text{Id}_{\IR^{n-1}}.
\end{equation}
We will now discuss standard properties of Floquet theory that are extensively used all along our analysis. In particular, we 
give a precise description 
of $W^u(\Lambda)$ and $W^s(\Lambda)$ in these local coordinates. We refer the reader to~\cite[Chapter 3]{teschl2012ordinary}  
for a detailed presentation on Floquet theory.

First of all, the fundamental solution $U$ satisfies the groupoid equation~:
\begin{equation}
U(\theta_1,\theta_2)U(\theta_2,\theta_3)=U(\theta_1,\theta_3)\text{ and }U(\theta,\theta)=\text{Id}_{\IR^{n-1}}.
\end{equation}
Note that $U(\mathcal{P}_{\Lambda}+\theta,\mathcal{P}_{\Lambda}+\theta_0)=U(\theta,\theta_0)$ but 
$U(\theta+\mathcal{P}_{\Lambda},\theta)\neq \text{Id}$ a priori. However
setting \begin{equation}
M(\theta)=U(\theta+\mathcal{P}_{\Lambda},\theta)
\end{equation}
yields the following Lemma:
\begin{lemm}
For all $\theta\in [0,\mathcal{P}_{\Lambda}]$, the monodromy matrix
$M(\theta)=U(\theta+\mathcal{P}_{\Lambda},\theta)$ depends smoothly in
$\theta$ and is $\mathcal{P}_{\Lambda}$--periodic.
Furthermore, all monodromy matrices are conjugated
\begin{eqnarray}
M(\theta+\mathcal{P}_{\Lambda})=M(\theta)\\
M(\theta_1)=U^{-1}(\theta_0,\theta_1)M(\theta_0)U(\theta_0,\theta_1).
\end{eqnarray}
\end{lemm}
\begin{proof}
$$U(\theta_1+\mathcal{P}_{\Lambda},\theta_1)=U(\theta_1+\mathcal{P}_{\Lambda},\theta_0+\mathcal{P}_{\Lambda})U(\theta_0+\mathcal{P}_{\Lambda},\theta_0)
U(\theta_0,\theta_1)  $$
but $U(\theta_1+\mathcal{P}_{\Lambda},\theta_0+\mathcal{P}_{\Lambda})=U(\theta_1,\theta_0)$
which yields the result.
\end{proof}
A direct consequence of that Lemma is that \emph{the spectrum of monodromy matrices $M(\theta)$ does not depend on $\theta$}. Its eigenvalues are called 
the Floquet multipliers of the closed orbit $\Lambda:=\{(0,\theta):\theta\in\IR/(\ml{P}_{\Lambda}\IZ)\}$. Saying that the closed orbit $\Lambda$ is 
\textbf{hyperbolic} is equivalent to the fact that no Floquet multipliers lies on the unit circle $\IS^1$. Equivalently, it also says that the 
monodromy matrices are hyperbolic. We shall now always suppose that the monodromy matrix is hyperbolic and diagonalizable in $\IC$. 
In that case, note that, up to a change of linear coordinates in $\IR^{n-1}$, we can suppose that $M:=M(0)$ takes the following form:
$$M=\text{Diag}(S_1,\ldots, S_p,U_{p+1},\ldots , U_q)\in GL_{n-1}(\IR),$$
where one has
\begin{itemize}
 \item for every $1\leq i\leq p$, one has $S_i=\nu_i\text{Id}_{\IR^1}$  with $|\nu_i|<1$ or $S_i=\nu_i R_{\vartheta_i}$ with $0<\nu_i<1$ and 
 $\displaystyle R_{\vartheta_i}:=\left(\begin{array}{cc}\cos(\vartheta_i)& -\sin(\vartheta_i)\\
\sin(\vartheta_i) & \cos(\vartheta_i) \end{array}\right)$,
  \item for every $p+1\leq i\leq q$, one has $S_i=\nu_i\text{Id}_{\IR^1}$ with $|\nu_i|>1$ or $S_i=\nu_i R_{\vartheta_i}$ with $\nu_i>1$ and 
 $\displaystyle R_{\vartheta_i}:=\left(\begin{array}{cc}\cos(\vartheta_i)& -\sin(\vartheta_i)\\
\sin(\vartheta_i) & \cos(\vartheta_i) \end{array}\right)$.
\end{itemize}
The \textbf{Lyapunov exponents} of the closed orbit are then given by the value $\frac{\log|\nu_i|}{\ml{P}_{\Lambda}}$ 
(which appear with multiplicity $2$ when they correspond 
to a complex eigenvalue). We denote\footnote{Some of them may be equal.} them by $(\chi_i(\Lambda))_{i=1,\ldots, n-1}$ 
and by $r$ the number of negative Lyapunov exponents. 
Note now that each matrix $\nu_i R_{\vartheta_i}$ can be put under an exponential form as 
follows~:
$$\nu_i R_{\vartheta_i}=\exp\left(\log(\nu_i)\text{Id}_{\IR^2}+\left(\begin{array}{cc}0& -\vartheta_i\\
\vartheta_i & 0\end{array}\right)\right).$$
For the blocks of size one, this will of course depend on the sign of $\nu_i$. In order to avoid that difficulty, we consider the matrix $M^2$ and we can 
then write 
$$M^2=\exp (2\ml{P}_{\Lambda}A_{\Lambda}),$$
where $A_{\Lambda}$ is a real valued matrix which has also a block-diagonal form given by the above matrices. For a general $\theta$, we set 
$\tilde{A}_{\Lambda}(\theta):=U(0,\theta)^{-1}A_{\Lambda}U(0,\theta)$ and we have the relation $M(\theta)^2=\exp (2\ml{P}_{\Lambda}\tilde{A}_{\Lambda}(\theta))$. 
Let us now state the main result we needed 
in our analysis and which follows from the Floquet Theorem~\cite{teschl2012ordinary}~:
\begin{prop} With the above notations, one has
\begin{equation}
U(\theta,\theta_0)=P(\theta,\theta_0)e^{(\theta-\theta_0)\tilde{A}_{\Lambda}(\theta_0)}
\end{equation}
$P(\theta,\theta_0)$ is a real $2\mathcal{P}_{\Lambda}$-periodic matrix depending smoothly on $\theta$. Furthermore,$P(\ml{P}_{\Lambda},0)$ is a
diagonal matrix with entries on the diagonal which are equal to $\pm 1$ where the term on the diagonal is equal to $-1$ 
whenever $\nu_i<0$ and to $+1$ otherwise.
\end{prop}
Note that we immediatly deduce from that proposition the following fact:
$$\tilde{A}_{\Lambda}(\theta)=P(\theta,0) A_{\Lambda}P(\theta,0)^{-1}.$$
Then, this proposition allows us to give a precise description of the unstable (resp. stable) manifold $W^u(\Lambda)$ (resp. $W^s(\Lambda)$) of $\Lambda$. 
Indeed, by construction, one can verify that
$$W^u(\Lambda)=\left\{\left(P(\theta_0,0)(0,y),\theta_0\right): (y,\theta_0)\in\IR^{n-1-r}\times\IR/(\ml{P}_{\Lambda}\IZ)\right\},$$
and 
$$W^s(\Lambda)=\left\{\left(P(\theta_0,0)(x,0),\theta_0\right): (x,\theta_0)\in\IR^{r}\times\IR/(\ml{P}_{\Lambda}\IZ)\right\}.$$
These submanifolds are invariantly fibered by the following smooth submanifolds:
$$\forall\theta_0\in\IR/(\ml{P}_{\Lambda}\IZ),\ W^{uu}(\theta_0):=\left\{\left(P(\theta_0,0)(0,y),\theta_0\right): y\in\IR^{n-1-r}\right\},$$
and
$$\forall\theta_0\in\IR/(\ml{P}_{\Lambda}\IZ),\ W^{ss}(\theta_0):=\left\{\left(P(\theta_0,0)(x,0),\theta_0\right): x\in\IR^{r}\right\}.$$
Note that $W^{u}(\Lambda)$ is non orientable if and only if
\begin{equation}
\det P(\mathcal{P}_{\Lambda},0)|_{\{0\}\oplus\IR^{n-1-r}}=-1.
\end{equation}
The same holds for the stable manifolds. In particular, observe that $\text{det}\ M$ is always positive thanks 
to the Liouville's formula~\cite[Lemma 3.11, p.83]{teschl2012ordinary}. 
This implies that $\det\ P(\mathcal{P}_{\Lambda},0)=1$ and thus 
$W^{u}(\Lambda)$ is orientable if and only if $W^{s}(\Lambda)$ is orientable.

\end{document}